\documentclass[
a4paper,
reprint,
amsmath,
amssymb,
superscriptaddress,
onecolumn,
accepted=2024-09-06]
{quantumarticle}
\pdfoutput=1
\usepackage[numbers,sort&compress,square]{natbib}

\usepackage{graphicx}
\usepackage{dcolumn}
\usepackage{bm}
\usepackage[colorlinks]{hyperref}
\usepackage{mathrsfs}
\usepackage{braket}
\usepackage{mathtools}
\usepackage{autobreak}
\usepackage{algorithmic}
\usepackage{algorithm}
\usepackage{color}
\usepackage{amsthm}
\usepackage{quantikz}
\usepackage{empheq}
\usepackage[caption=false]{subfig}

\newtheorem{result}{Result}

\newtheorem{lem}{Lemma}

\DeclareMathOperator{\diag}{diag}

\begin{document}

\title{End-to-end complexity for simulating the Schwinger model on quantum computers}
\author{Kazuki Sakamoto}
\email{kazuki.sakamoto.osaka@gmail.com}
\affiliation{%
Graduate School of Engineering Science, Osaka University\\
1-3 Machikaneyama, Toyonaka, Osaka 560-8531, Japan
}

\author{Hayata Morisaki}
\affiliation{%
Graduate School of Engineering Science, Osaka University\\
1-3 Machikaneyama, Toyonaka, Osaka 560-8531, Japan
}

\author{Junichi Haruna}
\affiliation{%
Center for Quantum Information and Quantum Biology, Osaka University 560-0043, Japan.
}

\author{Etsuko Itou}
\affiliation{Center for Gravitational Physics, Yukawa Institute for Theoretical Physics, Kyoto University, Kitashirakawa Oiwakecho, Sakyo-ku, Kyoto 606-8502, Japan}
\affiliation{Interdisciplinary Theoretical and Mathematical Sciences Program (iTHEMS), RIKEN, Wako 351-0198, Japan}

\author{Keisuke Fujii}
\affiliation{%
Graduate School of Engineering Science, Osaka University\\
1-3 Machikaneyama, Toyonaka, Osaka 560-8531, Japan
}
\affiliation{%
Center for Quantum Information and Quantum Biology, Osaka University 560-0043, Japan.
}
\affiliation{Center for Quantum Computing, RIKEN, Hirosawa 2-1, Wako Saitama 351-0198, Japan}

\author{Kosuke Mitarai}
\email{mitarai.kosuke.es@osaka-u.ac.jp}
\affiliation{%
Graduate School of Engineering Science, Osaka University\\
1-3 Machikaneyama, Toyonaka, Osaka 560-8531, Japan
}
\affiliation{%
Center for Quantum Information and Quantum Biology, Osaka University 560-0043, Japan.
}

\begin{abstract}
{
The Schwinger model is one of the simplest gauge theories. 
It is known that a topological term of the model leads to the infamous sign problem in the classical Monte Carlo method. 
In contrast to this, recently, quantum computing in Hamiltonian formalism has gained attention.
In this work, we estimate the resources needed for quantum computers to compute physical quantities that are challenging to compute on classical computers. 
Specifically, we propose an efficient implementation of block-encoding of the Schwinger model Hamiltonian. 
Considering the structure of the Hamiltonian, this block-encoding with a normalization factor of $\mathcal{O}(N^3)$ can be implemented using $\mathcal{O}(N+\log^2(N/\varepsilon))$ T gates.
As an end-to-end application, we compute the vacuum persistence amplitude. 
As a result, we found that for a system size $N=128$ and an additive error $\varepsilon=0.01$, with an evolution time $t$ and a lattice spacing $a$ satisfying $t/2a=10$, the vacuum persistence amplitude can be calculated using about $10^{13}$ T gates. 
Our results provide insights into predictions about the performance of quantum computers in the FTQC and early FTQC era, clarifying the challenges in solving meaningful problems within a realistic timeframe.
}
\end{abstract}

\maketitle

%--------------------------------------------------

\section{Introduction}\label{sec1}

Gauge theories capture the laws governing fundamental particles in the natural world. 
The Schwinger model~\cite{Schwinger:1962tp}, which represents quantum electrodynamics in 1+1 dimensions, is one of the simplest yet non-trivial gauge theories, making it a prime candidate for computational studies. 

Theoretical methods and numerical simulations have revealed the interesting phenomena of the Schwinger model, for instance, confinement, chiral symmetry breaking, and dynamical energy-gap generation.
Theoretical approaches such as bosonization, mass perturbation theory, and heavy-mass limit can make clear some properties of the Schwinger model but for only a limited range of the model parameters~\cite{Lowenstein:1971fc, Casher:1974vf, Coleman:1975pw, Coleman:1976uz, Manton:1985jm, Hetrick:1988yg, Jayewardena:1988td, Sachs:1991en, Adam:1993fc, Adam:1994by, Hetrick:1995wq, Narayanan:2012du, Narayanan:2012qf, Lohmayer:2013eka, Tanizaki:2016xcu}.
For wider mass and coupling regimes, the classical Monte Carlo method based on importance sampling is the most popular tool to study the quantitative properties of the model.
However, it is known that the Monte Carlo method suffers from the notorious sign problem if we incorporate the topological $\theta$-term in the model or consider a real-time evolution of its dynamics~\cite{Fukaya:2003ph, Nagata:2021ugx}.

It, therefore, motivates us to use quantum computers, which can simulate general quantum systems with a polynomial computational resource concerning their system sizes~\cite{feynman1982simulating, lloyd1996universal}.
From the perspective of complexity theory, it has been demonstrated that quantum computers are necessary to efficiently perform certain calculations in high energy physics \mbox{\cite{jordan2018bqp}}, and quantum algorithms for high energy physics have been studied extensively \mbox{\cite{jordan2012quantum, jordan2014quantum, roggero2020quantum, watson2023quantum, rubin2024quantum, byrnes2006simulating, klco2019digitization, yeter2019scalar, barata2021single, kan2021lattice, Tong2022provablyaccurate, klco20202, li2023simulating, bauer2023quantum, davoudi2023general, atas2023simulating, Grabowska2024vop}}.
As a simple and essential model of gauge theories,
the Schwinger model has been used as a benchmark to test new computational methods using quantum computation algorithms or tensor network methods~\cite{Banuls:2013jaa, Banuls:2015sta, Banuls:2016lkq,  Buyens:2016ecr, Buyens:2016hhu, Banuls:2016gid, ercolessi2018phase, klco2018quantum, Funcke:2019zna, Chakraborty:2020uhf, Kharzeev:2020kgc, magnifico2020real, avkhadiev2020accelerating, Honda:2021aum, Honda:2021ovk, deJong:2021wsd, Nguyen:2021hyk, Honda:2022edn, thompson2022quantum, xie2022variational, Nagano:2023uaq, Itou:2023img, mueller2023quantum, farrell2024scalable}.
It is customary to use the Hamiltonian formulations of the lattice Schwinger model to simulate the dynamics on quantum computers, and in the Hamiltonian formulations, the sign problem is absent from the beginning.
There are two types of Hamiltonian formulations so far. One is obtained by mapping the electron and positron on a staggered fermion lattice and the electric field on its links \cite{Kogut:1974ag}.
Another is the one obtained by eliminating the degree of freedom of the electric field with Gauss's law \cite{Banks1976, Hamer1997}.

The former is advantageous in that the interactions are geometrically local.
This property leads to the state-of-the-art algorithm by Tong \textit{et al.} that implements $e^{-iHt}$ in time $\widetilde{\mathcal{O}}\left(N t\ \mathrm{polylog}\left(\Lambda_0 \varepsilon^{-1}\right)\right)$, where $N$ is the number of fermionic sites, $\Lambda_0$ is the initial magnitude of the electric field at local links, and $\varepsilon$ is the error measured in operator norm \cite{Tong2022provablyaccurate}.
It is achieved by various sophisticated techniques, including qubitization by a linear combination of unitaries \cite{Low2019hamiltonian}, Hamiltonian simulation in the interaction picture \cite{LowWiebe2018}, and decomposition of the time evolution operator to local blocks using the Lieb-Robinson bound~\cite{haah2018}.
However, this formulation needs more qubits to merely express the system itself than the latter one, due to the need for simulating electric degrees of freedom.
Also, to achieve the $\widetilde{\mathcal{O}}\left(N t\mathrm{polylog}\left(\Lambda_0 \varepsilon^{-1}\right)\right)$ scaling, we would need many ancillary qubits to implement the assumed HAM-T oracles for Hamiltonian simulation in the interaction picture~\cite{LowWiebe2018}.
Trotterization \cite{Shaw2020quantumalgorithms} and its hybridization with qubitization techniques \cite{Rajput2022hybridizedmethods} can conduct the simulation without such ancillary qubits for oracles and with a reduced number of them, respectively.
However, they come at the cost of increased runtime $\widetilde{\mathcal{O}}\left(N^{5 / 2} t^{3 / 2}/\varepsilon^{1 / 2}\right)$ and $\widetilde{\mathcal{O}}\left(N^2 t^2/\varepsilon\right)$, respectively, which can be troublesome, especially for long-time simulations.

Hamiltonians in the latter formulation remove the degree of freedom of the electric field and alleviates the number of qubits needed to express the system, but at the cost of all-to-all interaction on the fermions \cite{Banks1976,Hamer1997}.
Removing the electric field from the Hamiltonian makes it hardware-friendly and helps the experimental realizations on current devices \cite{martinez2016realtime}.
This formulation has been studied mainly via Trotterization, possibly because of its experimental ease.
Nguyen \textit{et al.} \cite{Nguyen2022} have shown that we can achieve $\mathcal{O}\left(N^{4+1 / p} t^{1+1 / p}/\varepsilon^{1/p}\right)$ cost by utilizing $p$-th order product formula.
In this work, we use this formulation to seek feasibility on early fault-tolerant quantum computers which have a limited number of qubits. 

In the field of quantum algorithms, Trotterization can be superior to post-Trotter algorithms (e.g. quantum eigenvalue transformation based on block-encoding \mbox{\cite{Low2019hamiltonian, chakraborty2018power, gilyen2019quantum}}) both asymptotically and non-asymptotically in certain cases. 
For example, in the case of the Fermi-Hubbard model, one can achieve Trotter error that scales with the number of fermions rather than the number of fermionic modes by using physical constraints and commutator scaling \mbox{\cite{clinton2021hamiltonian, campbell2021early, schubert2023trotter}}. 
This feature has also been applied to several Hamiltonians in HEP whose structures are similar to that of the Fermi-Hubbard model \mbox{\cite{watson2023quantum}}.
Additionally, Trotterization can achieve the nearly optimal gate complexity for simulating geometrically local lattice systems \mbox{\cite{childs2019nearly}}, and can achieve the lowest asymptotic gate cost for simulating electronic structure Hamiltonians in second quantization in the plane-wave basis \mbox{\cite{su2021nearly}}. 
However, the Schwinger model cannot benefit from these advantages (mainly due to the large norm and all-to-all connectivity). 
In fact, even if we use the state-of-the-art theory (Theorem 1 in Ref. \mbox{\cite{su2021nearly}}), the total gate complexity is $\mathcal{O}\left(N^{4+1 / p} t^{1+1 / p}/\varepsilon^{1/p}\right)$ considering that the number of Trotter steps is $\mathcal{O}\left(N^{2+1 / p} t^{1+1 / p}/\varepsilon^{1/p}\right)$ and the gate complexity of each Trotter step is $\mathcal{O}\left(N^2\right)$ when we use $p$-th order product formula. 
On the other hand, the Hamiltonian of the Schwinger model has a distinctive structure, which is expected to allow us to implement efficient block-encodings.
These facts motivate us to construct block-encoding based algorithms as described below.

In this work, we develop quantum circuits for the so-called block-encoding \cite{gilyen2019quantum} of the Schwinger model Hamiltonian in the latter formulation.
While we focus on the Schwinger model, the techniques introduced in this work have the possibility to be applied to Hamiltonians with all-to-all interactions and a constant number of parameters (the Schwinger model has only four parameters, and this number is independent of the system size).
Utilizing the techniques developed in Refs.~\cite{babbush2018encoding, von2021quantum} and optimizing it for this particular model, we obtain a block-encoding with $\mathcal{O}(N+\log^2(N/\varepsilon))$ T gates and $\mathcal{O}(\log(N))$ ancillary qubits, and having the normalization factor that scales $\mathcal{O}(N^{3})$.
This provides us a quantum algorithm to simulate $e^{-iHt}$ with $\mathcal{O}(N^4t + (N^3 t + \log(1/\epsilon)) \log^2(Nt/\varepsilon) )$ 
T gates, using the quantum eigenvalue transformation algorithm \cite{gilyen2019quantum}.
Importantly, we give its concrete resource including every constant factor, which is described in Table \ref{tab:summary-block-encoding}. 
Moreover, we give an \textit{end-to-end} complexity for computing the quantity called vacuum persistence amplitude defined as $|\braket{\mathrm{vac}|e^{-iHt}|\mathrm{vac}}|$~\cite{PhysRev.82.664, Muschik:2016tws, Martinez:2016yna}.
The vacuum persistence amplitude observes one of the most interesting quantum effects, since it describes the instability of vacuum due to quantum fluctuations.
We set $\ket{\mathrm{vac}}$ as a computational basis state (N\'{e}el state) which corresponds to a ground state of the Hamiltonian under a certain parameter setting. 
This is in high contrast with the previous estimates on quantum chemistry and condensed matter problems \cite{babbush2018encoding,lee2021even,childs2017quantum,Yoshioka2022}, which neglects the cost for preparing the approximate ground state of the system.
We find that for $N=128$ and an additive error $\varepsilon=0.01$, with an evolution time $t$ and a lattice spacing $a$ satisfying $t/2a=10$, $|\braket{\mathrm{vac}|e^{-iHt}|\mathrm{vac}}|$ can be calculated with $1.72\times 10^{13}$ T gates.
Our work paves the way toward the practical quantum advantage.

This paper is organized as follows.
In Sec.~\ref{seing, wc2}, we briefly review the definition and physical properties of the Schwinger model.
We also review the block encoding and its application to quantum simulations of quantum systems.
Sec.~\ref{sec3} is the main part of our paper. There, we estimate the cost with the block encoding of the Schwinger model Hamiltonian.
In Sec.~\ref{sec:end-to-end-resource-estimates}, we provide the resource estimates for calculating the vacuum persistence amplitude. 
Moreover, we compare the complexities of our approach and the prior one. 
Sec.~\ref{sec:conclusion} is devoted to the conclusions.

\section{Preliminaries}\label{seing, wc2}

\subsection{Schwinger model}\label{sec:Schwinger model}
\subsubsection{Lagrangian, Hamiltonian, and qubit description}
In this section, we briefly review the Schwinger model in the continuum spacetime and on a lattice, and also its qubit description, following Refs.~\cite{Honda:2021aum, Honda:2021ovk}.
The Schwinger model is quantum electrodynamics in $1+1$ dimensions and its Lagrangian in continuous spacetime is given by
\begin{align}
    L = \int dx\, \biggl(\bar{\psi}(x) (i\gamma^\mu D_\mu-m)\psi(x) 
    -
    \frac{1}{4} F^{\mu\nu}(x)F_{\mu\nu}(x)
    +\frac{\theta}{2\pi} \epsilon^{\mu\nu} \partial_\mu A_\nu(x)
    \biggr),
\end{align}
where $\psi$ is a Dirac field which represents the electron.
$D_\mu\psi(x)$ and $F_{\mu\nu}$ are defined as
\begin{align}
    D_\mu \psi(x) &\coloneqq (\partial_\mu - ig A_\mu(x))\psi(x),\\
    F_{\mu\nu}(x) &\coloneqq \partial_\mu A_\nu(x) - \partial_\nu A_\mu(x),
  \quad  F^{\mu\nu}(x)  \coloneqq  \eta^{\mu\rho}\eta^{\nu\sigma}F_{\rho\sigma},
\end{align}
where $A_\mu(x)$ is a $U(1)$ gauge field which represents the photon field, and the temporal derivative of $A_1$ corresponds to the electric field.
$\gamma^\mu$ is the Dirac gamma matrix and $\epsilon^{\mu\nu}$ is the completely anti-symmetric tensor.
The metric $\eta^{\mu\nu}$ in the Minkowski spacetime is defined as $\eta^{\mu\nu} \coloneqq \diag (-1,1)$.
This model has a gauge redundancy and is invariant under the $U(1)$ gauge transformation.
To reduce the gauge redundancy, it is common to introduce a gauge-fixing condition. Here, we choose the temporal gauge, $A_0=0$.
The model has three real parameters, $g$, $m$, and $\theta$, which are respectively called the electric charge, the electron mass, and the coupling of the topological $\theta$-term. In the conventional lattice Mante Carlo method, the topological $\theta$-term causes the sign problem since this term takes a complex value in Euclidean spacetime.

After the Legendre transformation under these conditions, the Hamiltonian is given by
\begin{align}
\label{e:ContinuumSchwingerModelHamiltonian}
    H = \int dx\, \biggl(
    \frac{1}{2}\biggl(
    \Pi(x)-\frac{g\theta}{2\pi}
    \biggr)^2
    -i\bar{\psi}(x)\gamma^1(\partial_1 -igA_1(x))\psi(x) + m \bar{\psi}(x)\psi(x)
    \biggr),
\end{align}
where $\Pi(x)$ denotes  the conjugate momentum of $A_1(x)$, defined as 
    $\Pi(x) \coloneqq \partial_0 A_1(x)+ \frac{g\theta}{2\pi}.$
Furthermore, we impose the Guass' law to project the Hilbert space into the physical one
\begin{align}
\label{e:GaussLaw} 
    \partial_1 \Pi (x) - g \bar{\psi}(x) \gamma^0\psi(x)=0.
\end{align}

Now, we formulate this model on a lattice.
The photon field $A_1(x)$ is replaced by the link variable, $U_n \leftrightarrow \exp(ia g A_1(x))$, where $x\coloneqq na$ and $n=0,\ldots,N-1$ represents the position of the lattice.
The link variable is defined on the link between the $n$-th and $(n+1)$-th sites, and the conjugate momentum is replaced by $L_n \leftrightarrow - \frac{\Pi(x)}{g}$ on the $n$-th site.
To deal with the Dirac fermion on the lattice, we utilize the staggered fermion~\cite{Kogut:1974ag}.
The staggered fermion $\chi_n$ with lattice spacing $a$ represents the discretization of the two-component Dirac fermion, namely up and down spin components, with the lattice spacing $2a$.
Then, we replace $\psi(x)$ by
\begin{equation}
\psi(x)=\begin{pmatrix}
    \psi_{u}(x)\\
    \psi_{d}(x)
\end{pmatrix}
\leftrightarrow \frac{1}{\sqrt{2a}}
\begin{pmatrix}
    \chi_{\lfloor n/2 \rfloor}\\
    \chi_{\lfloor n/2 \rfloor+1}
\end{pmatrix}. 
\label{eq:chi_to_psi}
\end{equation}
Thus, the number of original Dirac fermion is $N/2$ on $N$-site lattice. 
We set $N$ to be an even number to consider the system that includes an integer number of Dirac fermions in this work.

The lattice version of the Schwinger model Hamiltonian \eqref{e:ContinuumSchwingerModelHamiltonian} is given by
\begin{align}
    H = J \sum_{n=0}^{N-2} \biggl(
    L_n+\frac{g\theta}{2\pi}
    \biggr)^2
    -iw \sum_{n=0}^{N-2} (
    \chi_n^\dagger U_n \chi_{n+1}
    - \chi_{n+1}^\dagger U_n^\dagger \chi_{n}
    ) 
    + m \sum_{n=0}^{N-1} (-1)^n \chi_n^\dagger \chi_n,
\end{align}
where $J\coloneqq g^2 a/2$ and $w \coloneqq 1/2a$.
Also, the lattice version of Gauss's law constraint is  
\begin{align}
L_n - L_{n-1} = \chi_n^\dagger \chi_n - \frac{1-(-1)^n}{2}.
\end{align}
In this work, we choose the formulation that solves Gauss's law imposing the boundary condition $L_{-1}=0$ and the gauge condition $U_n=1$.
Then, we can express the photon field $L_n$ in terms of the electron fields $\chi_n$ as
\begin{align}
    L_n = \sum_{i=0}^n \biggl(
    \chi_i^\dagger \chi_i - \frac{1-(-1)^i}{2}
    \biggr) .
\end{align}
Finally, using the Jordan-Wigner transformation \cite{Jordan:1928wi},
\begin{align}
    \chi_n \to \frac{X_n-iY_n}{2} \prod_{j=0}^{n-1} (-iZ_j),
\end{align}
the qubit description of the Schwinger model Hamiltonian is given by
\begin{equation}
\label{eq:schwinger}
H_S = J \sum_{n=0}^{N-2} 
\left(
    \sum_{i=0}^{n}\frac{Z_i + (-1)^i}{2} + \frac{\theta}{2\pi}
\right)^2
    + \frac{w}{2}\sum_{n=0}^{N-2}
\left(
    X_n X_{n+1} + Y_n Y_{n+1}
\right)
+ \frac{m}{2}\sum_{n=0}^{N-1} (-1)^n Z_n.
	\end{equation}
 We can see that the insertion of $\theta$-term does not induce any difficulty of the calculation since it gives a constant shift and $Z_i$ terms to the Hamiltonian.

Now, it is worth comparing this Hamiltonian with ones used in previous works about resource estimation. 
In the condensed matter problems, the Heisenberg model and the Hubbard model are often discussed~\cite{Yoshioka2022}. 
While the Hamiltonians of these models have only local interactions, the Schwinger model Hamiltonian which we investigate has all-to-all interactions, making our problem more challenging than typical condensed matter physics scenarios. 
On the other hand, the electronic Hamiltonian in quantum chemistry problems has even more complex all-to-all interactions. 
This results in a significantly large 1-norm of Hamiltonian coefficients (n.b. this value is an important element of complexity, see Sec.~\ref{subsec:block-encoding}~and~\ref{subsec:hamiltonian_simulation_intro}). 
However, using tensor factorization techniques such as tensor hypercontraction and performing numerical analysis, it has been observed that the norm scales as between $O(N)$ and $O(N^3)$ \cite{lee2021even}. 
This scaling may be smaller than the $O(N^3)$ scaling of the Schwinger model. 
It is noteworthy that the Schwinger model Hamiltonian in Eq.~(\ref{eq:schwinger}) has a similar form to the double low rank factorized electronic Hamiltonian \cite{lee2021even,von2021quantum}. 
In summary, the Hamiltonian we investigate is computationally more challenging than one of condensed matter physics and is comparable to or more difficult than the electronic Hamiltonian in quantum chemistry.

\subsubsection{Vacuum persistence amplitude}
In this work, we investigate the cost for calculating the vacuum persistence amplitude~\cite{PhysRev.82.664, Muschik:2016tws, Martinez:2016yna}.
It expresses the vacuum instability due to quantum fluctuations, which is defined as
\begin{equation}
    \mathcal{G}(t) = \langle \mathrm{vac}| e^{-i H_S t } |\mathrm{vac} \rangle. \label{eq:def-persistent}
\end{equation}
In this work, we set the N\'{e}el state as a vacuum,  $|\mathrm{vac} \rangle =| 1 0 1 0 \cdots \rangle$ in a computational basis, which realizes the ground state for $m \gg g$ region.

It is worth putting a comment on the physical meaning of this quantity. The Schwinger model is one of the relativistic theories, so that the vacuum instability causes the particle-antiparticle pair creation (and pair annihilation) via $\Delta E = \Delta m c^2$. 
The vacuum persistence amplitude is related to the particle production density defined as
\begin{align}
     \nu (t) &= \frac{1}{N} \sum_{n} \langle  \mathrm{vac}| (-1)^n    \hat{\psi}^\dag  \hat{\psi} (t)   | \mathrm{vac} \rangle \\
     &= \frac{1}{2N} \sum_{n} \langle  \mathrm{vac}| \left( (-1)^n  \hat{Z}_n (t)  +1  \right) | \mathrm{vac} \rangle 
\end{align}
as $\nu(t)=- N^{-1} \log ( |\mathcal{G}(t)|^2 )$ in the continuum limit. Here, $\hat{\psi}^\dag\hat{\psi} (t)$ and $\hat{Z}_n (t)$ denote the particle number operators at time $t$ in the Heisenberg picture. During a real-time evolution, this quantity oscillates due to quantum fluctuations. The vacuum persistence amplitude is useful to investigate such a dynamical quantum effect.

	\subsection{Block-encoding}\label{subsec:block-encoding}
	In this work, we utilize the framework of block-encoding \cite{gilyen2019quantum, chakraborty2018power} to implement $H_S$ on quantum computers.
    We say that the $(n_H+l)$-qubit unitary $U$ is an $(\alpha,l,\varepsilon)$-block-encoding of a Hamiltonian $H$ if
	\begin{equation}
	\left\| H - \alpha \left(\bra{0^l}\otimes I\right) U \left(\ket{0^l} \otimes I \right) \right\| \leq \varepsilon .
	\end{equation}
	In this work, we encode a Hamiltonian which is represented as a sum of $L$ unitary operators, i.e., 
	\begin{equation}
	H=\sum_{i=0}^{L-1}a_i U_i \quad (a_i>0,\ U_i: \text{unitary}). 
	\end{equation}
	Such a Hamiltonian can be block-encoded via the linear combination of unitaries (LCU) algorithm.
    The LCU algorithm constructs the block encoding of a Hamiltonian $H=\sum_{i=0}^{L-1}a_i U_i$ with the so-called PREPARE operator and the SELECT operator \cite{babbush2018encoding}. 
    The PREPARE operator $P$ is an $l=\lceil \log_2{L} \rceil$-qubit unitary that acts on the initial state $\ket{0^l}$ as
	\begin{equation}
	P \ket{0^l} = \sum_{i=0}^{L-1} \sqrt{\frac{a_i}{\alpha}} \ket{i} ,
	\end{equation}
	where $\alpha = \|\bm{a}\|_1 = \sum_{i=0}^{L-1} |a_i|$ is a normarization factor.
    The SELECT operator $V$ is an $(n_H+l)$-qubit unitary and is defined as
	\begin{equation}
	V = \sum_{i=0}^{L-1} \ket{i}\bra{i} \otimes U_i + \left( I - \sum_{i=0}^{L-1} \ket{i}\bra{i} \right) \otimes I.
	\end{equation}
    $(P^\dag \otimes I)V(P\otimes I)$ is an $(\alpha,l,0)$-block-encoding of a LCU Hamiltonian $H$ since
	\begin{equation}
	(\bra{0^l}\otimes I) (P^\dag \otimes I)V(P\otimes I) (\ket{0^l} \otimes I) = \frac{1}{\alpha}\sum_{i=0}^{L-1}a_i U_i.
	\end{equation} 
	Given a block-encoding of a Hamiltonian $H$, we can use the quantum eigenvalue transformation (QET) algorithm \cite{gilyen2019quantum} to implement a matrix function $f(H)$ approximately.

    We can also take a linear combination of block-encoded Hamiltonians. 
    Let $U_k$ be a block-encoding of $H_k$, i.e.,
    \begin{equation}
    \alpha_k(\bra{0}\otimes I)U_k(\ket{0}\otimes I)=H_k.
    \end{equation}
    The block-encoding of $\sum_{k=0}^{K-1} a_k H_k$, where $a_k$ is a positive real coefficient, can be implemented via the use of  a $\lceil\log_2 K\rceil$-qubit operation $P_{\mathrm{LC}}$ such that
    \begin{equation}
    \label{eq:PLC}
    P_{\mathrm{LC}}\ket{0} = \sum_{k=0}^{K-1} \sqrt{\frac{a_k \alpha_k}{\sum_l a_l \alpha_l}}\ket{k}.
    \end{equation}
    Namely, $(P_{\mathrm{LC}}^\dagger\otimes I) (\sum_k\ket{k}\bra{k}\otimes U_k) (P_{\mathrm{LC}}\otimes I)$ is a block-encoding of $\sum_k a_k H_k$ \cite{gilyen2019quantum}.

    \subsection{Hamiltonian simulation}\label{subsec:hamiltonian_simulation_intro}
    Using $U_H$ that is a $(\alpha,l,\varepsilon/(2|t|))$-block-encoding of a Hamiltonian $H$, we can implement a block-encoding of a time evolution operator $e^{-iHt}$.
    QET can implement a $(1,l+2,\varepsilon)$-block-encoding of $e^{-iHt}$ by two-types of procedure:
    \begin{itemize}
        \item $6\alpha|t|+9\ln{(6/\varepsilon)}$-times usages of $U_H$ or its inverses
        \item $3$-times usages of controlled-$U_H$ or its inverses
    \end{itemize}
    The process is explained in detail in Appendix~\ref{appsubsec:hamsim-ampest}. 

    \subsection{Amplitude estimation}
	Another important algorithm in this work is amplitude estimation \cite{brassard2000quantum, rall2023amplitude}. 
    Given a unitary $U_\psi$ such that $U_\psi\ket{0} = \ket{\psi}$ and a projector $\Pi$, amplitude estimation algorithm outputs an estimate of $a=|\Pi \ket{\psi}|$. 
    The original algorithm proposed in Ref. \cite{brassard2000quantum} outputs an estimate $\hat{a}$ satisfying $|a-\hat{a}|\leq \varepsilon$ with success probability $\geq 1-\delta$ in $\mathcal{O}(\varepsilon^{-1}\log{(1/\delta)})$ queries to the reflections $2\Pi - I$ and $2\ket{\psi}\bra{\psi}-I$. 
    Recent works improved the query complexity in terms of constant factors \cite{suzuki2020amplitude,grinko2021iterative,rall2023amplitude}. 
    We employ Chebyshev amplitude estimation \cite{rall2023amplitude} which provides the smallest query complexity at present.

\section{Block-encoding of Schwinger model Hamiltonian}\label{sec3}
In this section, we present our strategy to block-encode the Schwinger model Hamiltonian, $H_S$, and its complexity. We first explain our strategy in Sec.~\ref{sec:block-encoding-overview}.
Here, we decompose $H_S$ into several parts. The block-encodings of each term are shown in Sec.~\ref{sec:block1}--\ref{sec:block3}. In Sec.~\ref{sec:full-circuit}, we show the full circuit construction.
The proofs of the following Results are given in Appendix~\ref{appsec:resource-estimates}.

\subsection{Result overview}\label{sec:block-encoding-overview}
\begin{table*}
\centering
\caption{Summary of number of T gates, ancilla qubits, and definitions of the symbols}
\label{tab:summary-block-encoding}
\begin{tabular}{p{2.8cm}p{11.2cm}}
\hline \hline
\textbf{Category} & \textbf{Expression} \\
Number of T gates 
& $20N + 4 d\left(8\left\lceil\log_2{\left(\frac{28d\alpha_S}{\varepsilon}\right)}\right\rceil + 8 \lceil \log_2 N \rceil +2C-2\right) +312 \left\lceil \log_2 \left(\frac{546 \alpha_S}{\varepsilon}\right)\right\rceil + 124 \lceil \log_2 N \rceil
 +38\lceil \log_2 N' \rceil + 38\lceil \log_2 N'' \rceil + f + 78 C + 216$ \\
Number of ancilla qubits & $6\lceil\log_2 N \rceil + \max{(2\lceil \log_2 N' \rceil + \lceil \log_2 (N'-1) \rceil, 3\lceil \log_2 N'' \rceil)} + 6$ \\
Definitions & $d$ is an odd integer such that $d \geq \sqrt{2} \ln{(2\sqrt{14\alpha_S/\varepsilon})}$. \\
& $C = 5+4 \log_2(1+\sqrt{2}) \approx 10.09$. \\
& $f = 16\eta + 8\mu + 8\eta' + 8\mu' + 16\eta'' + 256 \lceil \log_2 L \rceil + 32 \lceil \log_2 K \rceil+ 32 \lceil \log_2 L' \rceil + 32 \lceil \log_2 K' \rceil + 64\lceil \log_2 L'' \rceil $. \\
& $(\eta, L, \eta', L', \eta'', L'',\mu,K, \mu',K') = (z(N), r(N), z(N'), r(N'), z(N''), r(N''), z(N-1), r(N-1), z(N'-1), r(N'-1))$ where, for an integer $M$, $z(M)$ and $r(M)$ are integers such that $M=2^{z(M)} r(M)$.\\
& $N' = \lceil N/2 \rceil$, $N'' = \lfloor N/2 \rfloor$. \\
\hline \hline
\end{tabular}
\end{table*}

\begin{result}[Block-encoding of the Schwinger model Hamiltonian]\label{result:schwinger-block-encoding:main}
Let $N\geq8\in\mathbb{N}$ be a system size and let $H_S$ be the Schwinger model Hamiltonian given by Eq. (\ref{eq:schwinger}). 
Let $b=\lceil\log_2{N}\rceil$.
For any $\varepsilon>0$, we can implement an unitary $U_{H_S}$ which is an $(\alpha_S, 2b+3, \varepsilon)$-block-encoding of
\begin{equation}
H_{S,\mathrm{mod}}\equiv
H_S - \frac{J}{8}\sum_{l=1}^{N-1} l^2 - J\sum_{n=1}^{N-1}\left(\frac{1}{2}\frac{1+(-1)^{n-1}}{2}+\frac{\theta}{2\pi}\right)^2 ,
\end{equation}
where
\begin{equation}\label{eq3.1}
\alpha_S = w(N-1) + \frac{m}{2}N +  \frac{J\theta}{2\pi} \sum_{\substack{l=1 \\ l: \mathrm{even}}}^{N-1} l +\left(\frac{J\theta}{2\pi}+\frac{J}{2}\right) \sum_{\substack{l=1 \\ l: \mathrm{odd}}}^{N-1} l + \frac{J}{8}\sum_{l=1}^{N-1} l^2,
\end{equation}
using the cost listed in Table \ref{tab:summary-block-encoding}.
\end{result}
Note that we can ignore the additive term proportional to the identity since this difference can be easily modified after simulating the Hamiltonian. 
The error $\varepsilon$ is derived from PREPARE operators, especially single-qubit rotation gates and fixed-point amplitude amplification, as shown in Appendix~\ref{appsubsec:prepare}. 

We will describe the overview to obtain the complexity stated above. 
First, we rewrite the Hamiltonian as:
\begin{equation}\label{eq:schwinger-mod}
H_S = H_{XX} + H_{YY} + H_{Z} + H_{Z\text{even}} + H_{Z\text{odd}} + H_{Z^2} + J\sum_{n=1}^{N-1}\left(\frac{1}{2}\frac{1+(-1)^{n-1}}{2}+\frac{\theta}{2\pi}\right)^2,
\end{equation}
where
\begin{align}
\begin{split}
    H_{XX} &= \frac{w}{2}\sum_{n=0}^{N-2}X_n X_{n+1}, ~
H_{YY} = \frac{w}{2}\sum_{n=0}^{N-2} Y_n Y_{n+1},~
H_{Z} = \frac{m}{2}\sum_{n=0}^{N-1}(-1)^n Z_n, \\
H_{Z\text{even}} &= J\frac{\theta}{2\pi} \sum_{\substack{n=1 \\ n: \mathrm{even}}}^{N-1} \sum_{i=0}^{n-1} Z_i, ~
H_{Z\text{odd}} = J \left( \frac{1}{2}+\frac{\theta}{2\pi}\right) \sum_{\substack{n=1 \\ n: \mathrm{odd}}}^{N-1} \sum_{i=0}^{n-1} Z_i, ~
H_{Z^2} = \frac{J}{4}\sum_{n=1}^{N-1}\left(\sum_{i=0}^{n-1}Z_i\right)^2.
\end{split}
\end{align}
Note that $H_Z$, $H_{Z\text{even}}$ and $H_{Z\text{odd}}$ can be summarized into one Hamiltonian in the form of $\sum_{i} a_i Z_i$ where $a_i$ is an elementary function of the index $i$ and parameters $J$, $\theta$, $w$, $N$ and $m$.
It is therefore possible to use arithmetics along with techniques like black-box state preparation \cite{sanders2019black} to construct PREPARE operation.
However, it requires many ancilla qubits and T gate cost and thus we take an alternative approach.

\subsection{Block-encoding of \texorpdfstring{$H_{XX}$}{}, \texorpdfstring{$H_{YY}$}{}, and \texorpdfstring{$H_{Z}$}{}}
\label{sec:block1}
The block-encodings of $H_{XX}$, $H_{YY}$, and $H_{Z}$ can be implemented easily.
The PREPARE operators and SELECT operators for these Hamiltonians are given as follows, respectively: 
\begin{equation}\label{eqapp1.10}
P_{XX}\ket{0^b} = \mathrm{UNI}_{(N-1)}\ket{0^b} = \frac{1}{\sqrt{N-1}} \sum_{n=0}^{N-2}\ket{n}, 
\quad V_{XX}=\sum_{n=0}^{N-2} \ket{n}\bra{n}\otimes X_nX_{n+1} + \sum_{n=N-1}^{2^b-1} \ket{n}\bra{n}\otimes I, 
\end{equation}
\begin{equation}\label{eqapp1.11}
P_{YY}\ket{0^b} = \mathrm{UNI}_{(N-1)}\ket{0^b} = \frac{1}{\sqrt{N-1}} \sum_{n=0}^{N-2}\ket{n}, 
\quad V_{YY}=\sum_{n=0}^{N-2} \ket{n}\bra{n}\otimes Y_n Y_{n+1} + \sum_{n=N-1}^{2^b-1} \ket{n}\bra{n}\otimes I, 
\end{equation}
\begin{equation}\label{eqapp1.12}
P_{Z}\ket{0^b} = \mathrm{UNI}_{(N)}\ket{0^b} = \frac{1}{\sqrt{N}} \sum_{n=0}^{N-1}\ket{n}, 
\quad V_{Z}=\sum_{n=0}^{N-1} \ket{n}\bra{n}\otimes (-1)^n Z_n + \sum_{n=N}^{2^b-1} \ket{n}\bra{n}\otimes I, 
\end{equation}
where $b=\lceil\log_2 N\rceil$. 
$\mathrm{UNI}_{(N)}$ denotes the unitary operator preparing the uniform superposition state, which can be implemented with $O(\log{N})$ T gates \cite{sanders2020compilation}. 
The detailed cost to prepare a uniform superposition state can be found in Appendix~\ref{appsubsec:prepare}. 
Also, the SELECT operators can be implemented with $O(N)$ T gates \cite{babbush2018encoding}. 

\subsection{Block-encoding of \texorpdfstring{$H_{Z\text{even}}$}{} and \texorpdfstring{$H_{Z\text{odd}}$}{}}
\label{sec:block2}
Next, let us discuss the construction of $H_{Z\text{even}}$.
We can construct the block encoding of $H_{Z\text{even}}$ as 
\begin{equation}
(P_{Z\text{even}}^\dag \otimes I) V_{Z^2} (P_{Z\text{even}} \otimes I),
\end{equation}
with a PREPARE operator $P_{Z\text{even}}$ and a SELECT operator $V_{Z^2}$, defined as 
\begin{align}
\label{eqapp1.15}
P_{Z\text{even}} \ket{0^b}\ket{0^b}  
&= \frac{1}{\sqrt{\alpha_{S1}}} \sum_{\substack{n=1 \\ n: \mathrm{even}}}^{N-1} \ket{n} \otimes \sum_{i=0}^{n-1}\ket{i}, \\
\label{eqapp1.16}
V_{Z^2} &= I \otimes \sum_{i=0}^{N-1} \ket{i}\bra{i} \otimes Z_i + I \otimes \left( I - \sum_{i=0}^{N-1} \ket{i}\bra{i} \right) \otimes I,
\end{align}
where $\alpha_{S1} = \sum_{n=1,n: \mathrm{even}}^{N-1}n$ is a normalization factor.
Here, $P_{Z\text{even}}$ can be written as a product of two operators; $P_{Z\text{even}}=P_2 (P_{S1} \otimes I)$.  $P_{S1}$ is a $b$-qubit unitary defined as
\begin{equation}\label{eqapp1.13}
P_{S1}\ket{0^b} = \frac{1}{\sqrt{\alpha_{S1}}} \sum_{\substack{n=1 \\ n: \mathrm{even}}}^{N-1} \sqrt{n}\ket{n},
\end{equation}
while $P_2$ is a $2b$-qubit unitary given by
\begin{equation}\label{eqapp1.14}
P_2 = \sum_{n=1}^{N-1} \ket{n}\bra{n} \otimes \left( \frac{1}{\sqrt{n}}\sum_{i=0}^{n-1}\ket{i}\bra{0^b} + \dots \right) + \ket{0}\bra{0}\otimes W_0 + \sum_{n=N}^{2^b-1} \ket{n}\bra{n} \otimes W_n, 
\end{equation}
where $W_0, W_N, \dots W_{2^b-1}$ are some unitaries whose actions we do not have to consider.
The state preparations by $P_{S1}$ can be executed with $O(\log{N})$ T gates \cite{babbush2018encoding} and $P_2$ can be implemented using $O(\log{(1/\varDelta)}\log{(N/\varepsilon)})$ T gates where $\varepsilon$ and $\Delta$ are error parameters (see Appendix \ref{appsubsec:prepare} for concrete definitions, \cite{gilyen2019quantum, yoder2014fixed, rall2023amplitude}).
The block-encoding of $H_{Z\text{odd}}$ can be implemented by replacing $P_{S1}$ operator with $P_{S2}$ defined as 
\begin{equation}\label{eqapp1.18}
P_{S2}\ket{0^b} = \frac{1}{\sqrt{\alpha_{S2}}} \sum_{\substack{n=1 \\ n: \mathrm{odd}}}^{N-1} \sqrt{n}\ket{n},
\end{equation}
where $\alpha_{S2} = \sum_{n=1,n: \mathrm{odd}}^{N-1}n$.
We provide the detailed construction of these operators $P_{S1},P_{S2},$ and $P_{2}$ in Appendix~\ref{appsubsec:prepare}. 

\subsection{Block-encoding of \texorpdfstring{$H_{Z^2}$}{}}
\label{sec:block3}
Lastly, we construct the block-encoding of $H_{Z^2}$. 
The following Lemma~\ref{lemapp:chev-block} is useful to implement the square of a Hamiltonian. 
\begin{lem}[Consructing the square of the block-encoded operator]\label{lemapp:chev-block}
Let $b\in\mathbb{N}$, $H$ be a Hamiltonian satisfying $\|H\|\leq1$, and $R_0=2\ket{0^b}\bra{0^b}-I$. Suppose that $U$ is an $(1,b,0)$-block-encoding of $H$. Then $U^\dag (R_0\otimes I) U$ is an $(1,b,0)$-block-eocoding of $2H^2-I$. This is a special case of quantum eigenvalue transformation with the Chebyshev polynomial of the first kind \cite{childs2017quantum, von2021quantum}.  
\end{lem}
Observing that $(P_{2}^\dag \otimes I) V_{Z^2} (P_{2}\otimes I)$ is a controlled-block-encoding of $n^{-1} \sum_{i=0}^{n-1} Z_i$ for each $n$:
\begin{align}\label{eqapp1.20}
(\bra{n} \otimes \bra{0^b}\otimes I)(P_{2}^\dag \otimes I) V_{Z^2} (P_{2}\otimes I)(\ket{n} \otimes \ket{0^b}\otimes I)
=\frac{1}{n} \sum_{i=0}^{n-1} Z_i,
\end{align}
these terms can be squared using Lemma \ref{lemapp:chev-block} in parallel:
\begin{align}
(\bra{n} \otimes \bra{0^b}\otimes I)(P_{2}^\dag \otimes I) V_{Z^2} (P_{2}\otimes I)
(I\otimes R_0 \otimes I) (P_{2}^\dag \otimes I) V_{Z^2} (P_{2}\otimes I)(\ket{n} \otimes \ket{0^b}\otimes I)
=\frac{2}{n^2} \left(\sum_{i=0}^{n-1} Z_i\right)^2 - I.
\end{align}
Now, we take a linear combination of the above operator using the operator $P_{S3}$, 
\begin{equation}\label{eqapp1.19}
P_{S3}\ket{0^b} = \frac{1}{\sqrt{\alpha_{S3}}} \sum_{n=1}^{N-1} n\ket{n},
\end{equation}
where $\alpha_{S3} = \sum_{n=1}^{N-1}n^2$ ,which can be implemented with $O(\log{N})$ T gates (see Appendix~\ref{appsubsec:prepare}, \cite{sanders2019black}).
Utilizing $P_{S3}$ and Eq.~\eqref{eq:PLC}, we can take the linear combination of the operators $2n^{-2} \left(\sum_{i=0}^{n-1} Z_i\right)^2 - I$ with coefficients $n^2$.
The resulting operator that encodes $H_{Z^2}$ is as in Fig.~\ref{fig:H_Z_2_block}. 
\begin{figure}
\centerline{
\includegraphics[width=150mm, page=1]{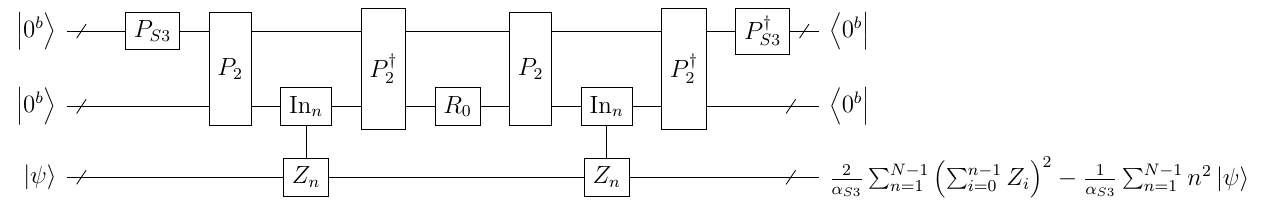}
}
\caption{The block-encoding of $H_{Z^2}$ term. }
\label{fig:H_Z_2_block}
\end{figure}

\subsection{Full circuit construction}
\label{sec:full-circuit}
Finally, we take a linear combination of unitaries that encode each term in Eq.~\eqref{eq:schwinger-mod} and construct the block-encoding.  
To this end, we introduce an operation $P_{\mathrm{split}}$ defined as
\begin{align}\label{eq:split}
\begin{split}
    P_{\mathrm{split}} \ket{000} = 
\frac{1}{\sqrt{\alpha_{S}}} &\left( \sqrt{\frac{w(N-1)}{2}}\ket{000} + \sqrt{\frac{w(N-1)}{2}} \ket{001} + \sqrt{\frac{mN}{2}} \ket{010} \right. \\
& +\sqrt{\frac{\alpha_{S1}J\theta_0}{2\pi}}  \ket{100} + \sqrt{\alpha_{S2}}\sqrt{\frac{J\theta_0}{2\pi}+\frac{ J}{2}} \ket{101} 
\left. + \sqrt{\frac{\alpha_{S3} J}{8}} \ket{110}\right). 
\end{split}
\end{align}
$\ket{000}$, $\ket{001}$, $\ket{010}$, $\ket{100}$, $\ket{101}$ and $\ket{110}$ branches are for $H_{XX}$, $H_{YY}$, $H_{Z}$, $H_{Z \text{even}}$, $H_{Z \text{odd}}$, and $H_{Z^2}$, respectively.
Defining $P_1$ as an operation that first applies $P_{\mathrm{split}}$ and then applies controlled versions of $\mathrm{UNI}_{(N-1)}$, $\mathrm{UNI}_{(N)}$, $P_{S1}$, $P_{S2}$, and $P_{S3}$ (see Appendix~\ref{appsubsec:prepare} for detail), we can implement the block-encoding of the Schwinger model Hamiltonian by the circuit in Fig.~\ref{fig:schwinger-block}.
\begin{figure}
\centerline{
\includegraphics[width=150mm, page=1]{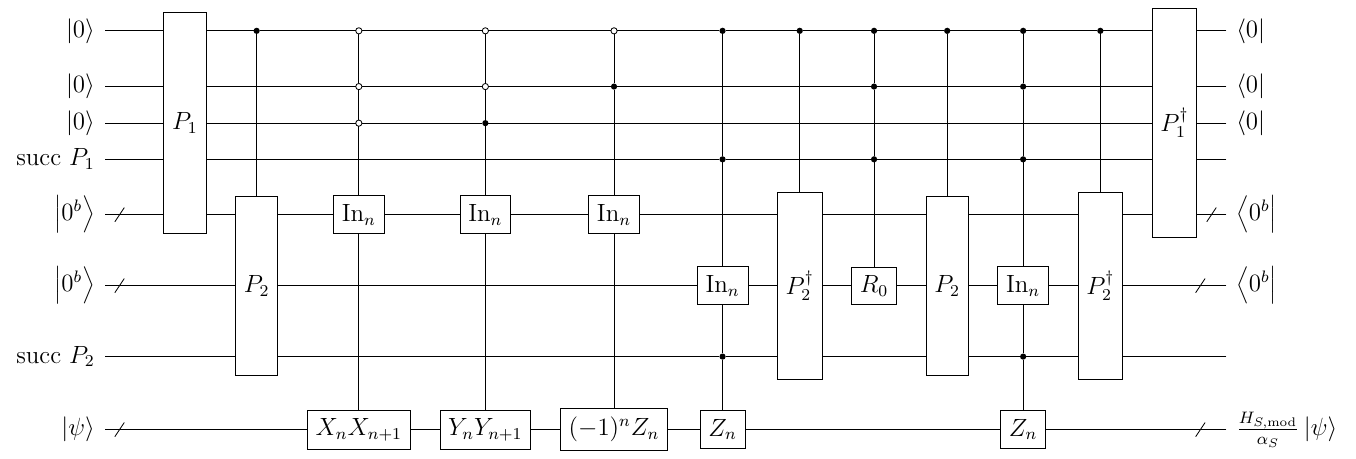}
}
\caption{Circuit for implementing the block-encoding of the Schwinger model Hamiltonian. $\mathrm{succ}\ P_1$ and $\mathrm{succ}\ P_2$ are single-qubit ancilla registers that flag success of the procedures, which enable us to perform an identity operator in the case of failure. }
\label{fig:schwinger-block}
\end{figure}

\section{End-to-end resource estimates}\label{sec:end-to-end-resource-estimates}

\subsection{Formulae for estimates}
Next, we give the result for the Hamiltonian simulation of the Schwinger model. 
The result can be obtained straightforwardly by using the quantum eigenvalue transformation introduced in \cite{gilyen2019quantum} to implement $e^{iH_St}$ from the block-encoding $U_{H_S}$.

\begin{result}[Simulating the Schwinger model Hamiltonian]\label{cor:simulation-schwinger-main}
Let $t\in \mathbb{R}$, $\varepsilon\in(0,1)$, $N\geq 8 \in \mathbb{N}$ be a system size, $b=\lceil \log_2{N} \rceil$, and $r$ be the smallest even number $\geq 2\alpha_S |t| + 3\ln{(9/\varepsilon)}$.
Let $\alpha_S$, $C$ and $H_{S,\mathrm{mod}}$ be as defined in Result~\ref{result:schwinger-block-encoding:main}.
Suppose that $U_{H_S}$ is a unitary which is an $(\alpha_S,2b+3,\varepsilon/(3|t|))$-block-encoding of $H_{S,\mathrm{mod}}$ and let $C_{H_S}$ be the number of T gates required to implement $U_{H_S}$ given in Result \ref{result:schwinger-block-encoding:main}. 
We can implement a unitary $U_{\mathrm{time}}$ which is an $(1,2b+5,\varepsilon)$-block-encoding of $e^{-i H_{S,\mathrm{mod}} t}$ using 
\begin{equation}
 r \left( 3C_{H_S} + 48 \left\lceil\log_2{\left(\frac{18(2r+1)}{\varepsilon} \right)}\right\rceil + 24 b + 12 C + 24 \right) \\
+ 3C_{H_S} + 24 \left\lceil\log_2{\left(\frac{18(2r+1)}{\varepsilon} \right)}\right\rceil + 40 b + 6C + 120
 \end{equation}
T gates and
\begin{equation}
 6b + \max{(2\lceil \log_2 N' \rceil + \lceil \log_2 (N'-1) \rceil, 3\lceil \log_2 N'' \rceil)} + 6
\end{equation}
ancilla qubits.
\end{result}

Finally, we give the cost for estimating the vacuum persistence amplitude. 
In this work, we use Chebyshev amplitude estimation~\cite{rall2023amplitude}. 
Let $\ket{\psi}$ be a state, $\Pi$ be a projector, and $\omega=\|\Pi \ket{\psi}\|$. 
Chebyshev amplitude estimation samples from a random variable $\hat{\omega}$ satisfying $\mathrm{Pr}[|\hat{\omega}-\omega|\geq \varepsilon]\leq \delta$ by calling reflection operators $2\Pi - I$ and $2\ket{\psi}\bra{\psi}-I$.
The number of queries to $2\Pi - I$ and $2\ket{\psi}\bra{\psi}-I$ in the algorithm, which we denote $Q_{\Pi}$ and $Q_{\psi}$ respectively, depends on $\omega$, $\varepsilon$, and $\delta$.
In Appendix~\ref{appsubsec:hamsim-ampest}, we show that $Q_{\Pi}+Q_{\psi}\approx 2000$ for $\delta=0.05$ and $\varepsilon=0.005$ by numerical simulation. 
Also, it has been shown in Ref.~\cite{rall2023amplitude} that $Q_{\Pi}\approx Q_{\psi}$ as the algorithm calls two reflection operators alternatively in quantum circuits.

To estimate the vacuum persistence amplitude, we set $\Pi = \ket{\text{vac}'}\bra{\text{vac}'}$, where $\ket{\text{vac}'} = \ket{0^{\nu-N}}\ket{\text{vac}}$ and $\nu = N+2\lceil \log_2{N} \rceil+5$, and implement $2\ket{\psi}\bra{\psi}-I$ via $U_{\text{time}}(2\ket{\text{vac}'}\bra{\text{vac}'}-I)U_{\text{time}}^\dagger$, where $U_{\text{time}}$ is the unitary which block-encodes $e^{-iH_{S,\text{mod}}t}$.
We therefore need two calls to $U_{\text{time}}$ and a single call of $2\ket{\text{vac}'}\bra{\text{vac}'}-I$ to implement $2\ket{\psi}\bra{\psi}-I$.
Noting that $2\ket{\text{vac}'}\bra{\text{vac}'}-I$ is equivalent to $2\ket{0^\nu}\bra{0^\nu}-I$ up to X gates and that $2\ket{0^\nu}\bra{0^\nu}-I$ can be implemented with $4\nu - 8$ T gates, we get the following result:

\begin{result}[Estimating the vacuum persistence amplitude, based on empirical assumption]\label{cor:vacuum-amplitude-estimate-main}
Let $t\in \mathbb{R}$, $N\geq 8 \in \mathbb{N}$ be a system size, $b=\lceil \log_2{N} \rceil$, $\ket{\mathrm{vac}} =\ket{1 0 1 0 \cdots }$, and
$C_{\mathrm{time}}$ be the number of T gates required to implement a $(1,2b+5,0.005)$-block-encoding of $e^{-i H_{S,\mathrm{mod}} t}$ via Result \ref{cor:simulation-schwinger-main}. 
Then we can sample from a random variable $\hat{G}$ satisfying $\mathrm{Pr}[|\hat{G}- |\braket{\mathrm{vac}|e^{-iH_S t}|\mathrm{vac}}| |\geq 0.01]\leq 0.05 $ using about $2000(C_{\mathrm{time}} + 4N + 8b + 12)$ T gates on average and $\max (N+2b+3, 6b + \max (2\lceil \log_2 N' \rceil + \lceil \log_2 (N'-1) \rceil, 3\lceil \log_2 N'' \rceil) + 6)$ ancilla qubits. 
\end{result}

A detailed analysis of this result is given in Appendix \ref{appsubsec:hamsim-ampest}.
Also, note that even considering constant factors, Chebyshev amplitude estimation is better than incoherent measurements (see Appendix \mbox{\ref{appsubsec:hamsim-ampest}} for detail).

\subsection{End-to-end T counts with realistic parameters}

\begin{table*}
\centering
\caption{T count and estimated runtime for calculating the vacuum persistence amplitude. The runtime assumes the T gate speed of 1~MHz.}
\label{tab:runtimes-end-to-end}
\begin{tabular}{lllll}
\hline \hline
system size & & & evolution time & \\
\hline
 & & $wt=1$ & $wt=10$ & $wt=100$ \\
\hline
$N=16$ & T count        & $9.11\times 10^{9}$ & $7.77\times 10^{10}$ & $8.25\times 10^{11}$ \\
     & runtime [days]   & $0.106$              & $0.899$               & $9.55$               \\
$N=32$ & T count        & $3.00\times 10^{10}$ & $3.25\times 10^{11}$ & $3.83\times 10^{12}$ \\
    & runtime [days]  & $0.347$              & $3.76$               & $44.3$                \\
$N=64$ & T count        & $1.88\times 10^{11}$ & $2.19\times 10^{12}$ & $2.54\times 10^{13}$ \\
     & runtime [days]  & $2.18$               & $25.3$               & $294$                \\
$N=128$ & T count       & $1.60\times 10^{12}$ & $1.72\times 10^{13}$ & $1.97\times 10^{14}$ \\
    & runtime [days]  & $18.5$               & $200$                & $2276$               \\
$N=256$ & T count       & $1.41\times 10^{13}$ & $1.61\times 10^{14}$ & $1.82\times 10^{15}$ \\
    & runtime [days]  & $163$                & $1864$               & $20990$              \\
\hline \hline
\end{tabular}
\end{table*}

Here, we estimate the runtime of the algorithm based on its T count using the results obtained in the previous section. 
We set the parameters as $a=0.2$, $m=0.1$, $\theta =\pi$, $g=1$, $w=1/(2a) = 2.5$, and $J=g^2 a/2=0.1$. 
Then, using Results~\ref{result:schwinger-block-encoding:main}--\ref{cor:vacuum-amplitude-estimate-main}, the runtimes to calculate the vacuum persistence amplitude with an additive error of $\varepsilon = 0.01$ with T gate consumption rate of 1~MHz,
which is also employed in the recent resource estimate ~\cite{Yoshioka2022}, are given as in Table \ref{tab:runtimes-end-to-end} and Fig.~\ref{fig:T-size-dep},~\ref{fig:T-time-dep}.

According to the results, in order to solve the problem within 100 days, the maximum size of the system that can be simulated is approximately $N=200$, $100$, and $40$ for $wt=1$, $10$, and $100$ respectively. 
Furthermore, in the case of calculating the vacuum persistence amplitude, we can say that the rate of 1~kHz is insufficient for solving the problem in a realistic timeframe, and the rate of 1~MHz is a minimum requirement, which is feasible considering the state-of-the-art magic state distillation protocol~\cite{litinski2019magic} and gate time.

\begin{figure}
\centering
\subfloat[Dependence on system size \( N \).]{
    \includegraphics[keepaspectratio, width=.4\linewidth]{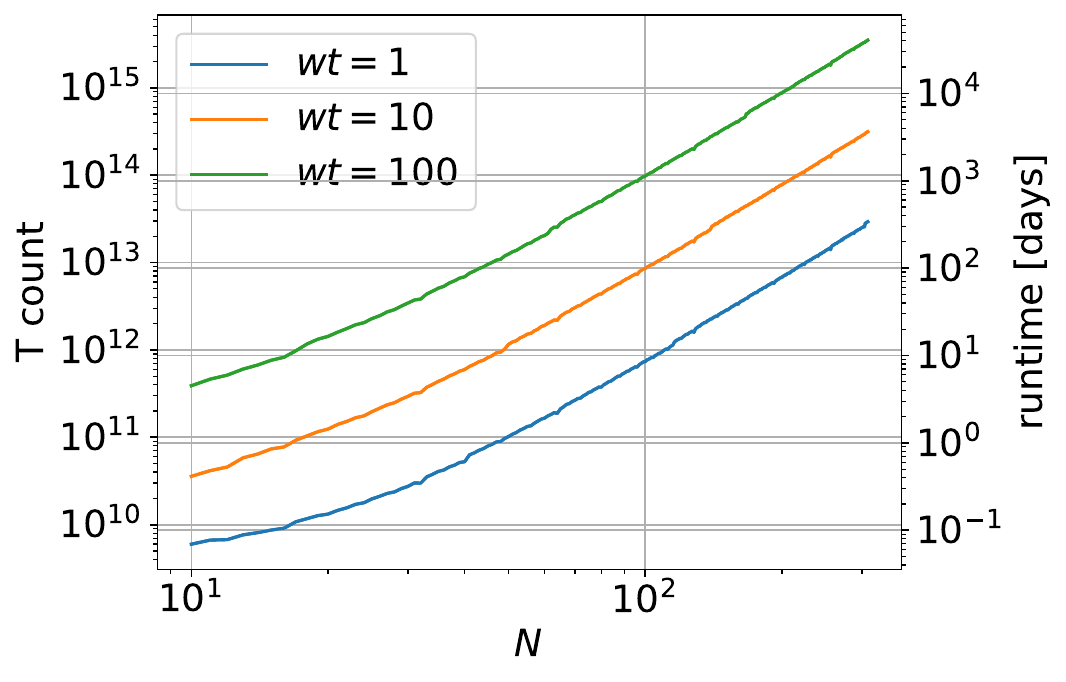}
    \label{fig:T-size-dep}
}
\hfil 
\subfloat[Dependence on evolution time \( t \).]{
    \includegraphics[keepaspectratio, width=.51\linewidth]{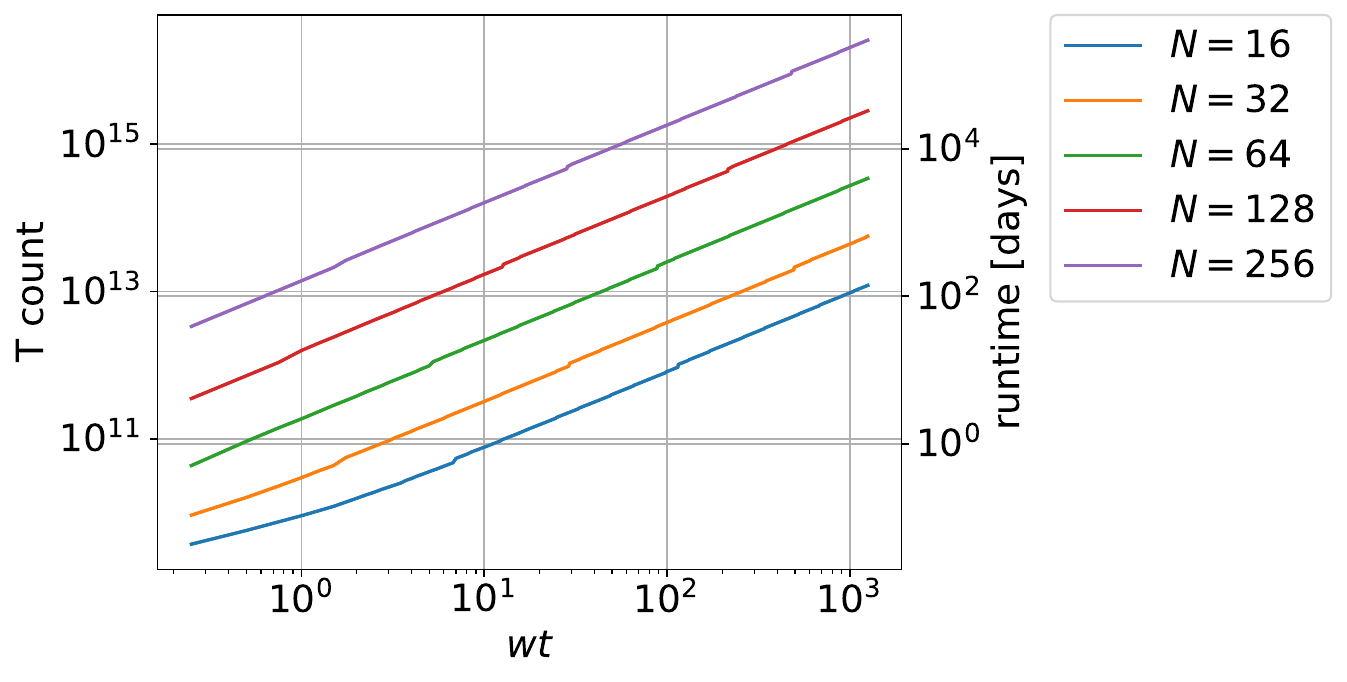}
    \label{fig:T-time-dep}
}
\caption{T counts for calculating the vacuum persistence amplitude, where the precision \( \varepsilon=0.01 \) is fixed.}
\label{fig:T-counts}
\end{figure}

\subsection{Comparison with the previous work}
Let us compare our results with the previous one by Shaw \textit{et al} \cite{Shaw2020quantumalgorithms}. 
Here, we consider the number of T gates needed to implement $e^{iH_S t}$ as Ref.~\cite{Shaw2020quantumalgorithms} only provides a formula for calculating the T count for this task. 
We set the parameters as $a=0.2$, $m=0.1$, $\theta =\pi$, $g=1$, $w=1/(2a) = 2.5$, and $J=g^2 a/2=0.1$. 
Note that we rescale the evolution time $T$ appeared in Ref.~\cite{Shaw2020quantumalgorithms} as $T\rightarrow 2^{-1}ag^2t$ since the Hamiltonian used in Ref.~\cite{Shaw2020quantumalgorithms} is non-dimensionalized by rescaling with a factor $2/(ag^2)$. 
Then, Fig.~\ref{fig:comp-T-size-dep}--\ref{fig:comp-T-precision-dep} denote T counts of our approach and the previous one.

Roughly speaking, T count of our approach is \( \tilde{\mathcal{O}}(N^4 t + N^3 t \log^3{(1/\varepsilon)}) \) and T count of the previous one is \( \tilde{\mathcal{O}}(N^{2.5}t^{1.5}/\varepsilon^{0.5}) \). 
Therefore, our approach has an advantage in the cost of simulating the Hamiltonian in the long-time and high-precision domain as shown in Fig.~\ref{fig:comp-T-time-dep} and \ref{fig:comp-T-precision-dep}. 
Note that, however, our approach needs a larger number of T gates than the previous one in the large-system domain as in Fig.~\ref{fig:comp-T-size-dep}. 
This disadvantage comes from the difference of the Hamiltonian formulation
rather than the choice between QSVT or Trotterization.

Finally, we would like to emphasize again that our approach uses fewer qubits.
Specifically, our approach requires $O(N)$ qubits while the previous work requires $O(N \log{N})$ qubits. 
This difference arises because the Hamiltonian formulation used in the previous work includes electric fields, requiring $O(\log{N})$ qubits to represent each electric field.
In contrast, the Hamiltonian formulation we chose removes the degrees of freedom of the electric fields, thereby eliminating the need for qubits to represent electric fields.

\begin{figure}
\centering
\subfloat[Dependence on system size \( N \). (\( wt=10 \), \( \varepsilon=0.01 \))]{%
    \includegraphics[keepaspectratio, width=.3\linewidth]{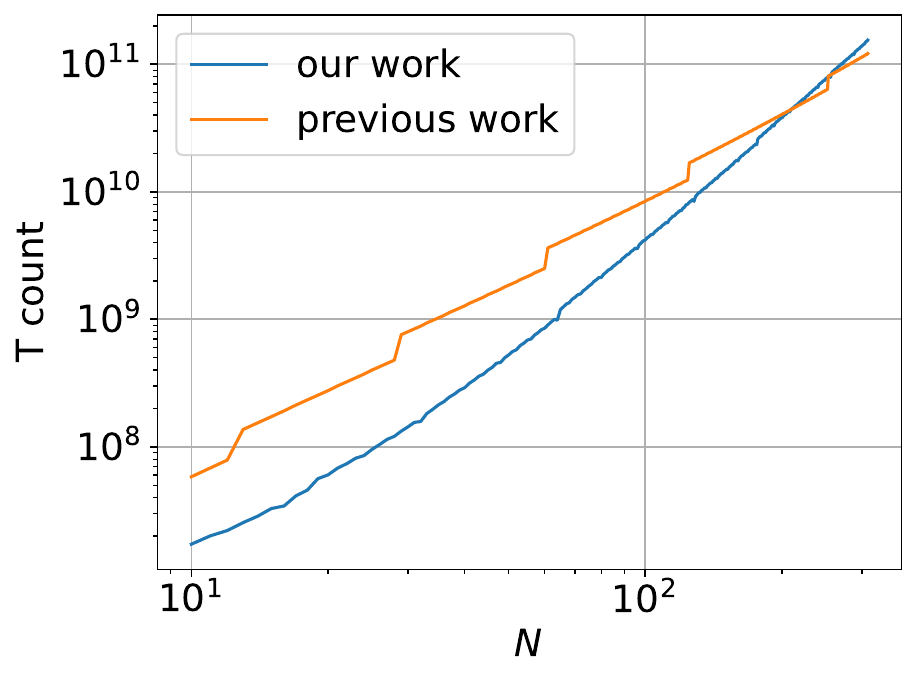}
    \label{fig:comp-T-size-dep}
}
\hfil 
\subfloat[Dependence on evolution time \( t \). (\( N=64 \), \( \varepsilon=0.01 \))]{%
    \includegraphics[keepaspectratio, width=.3\linewidth]{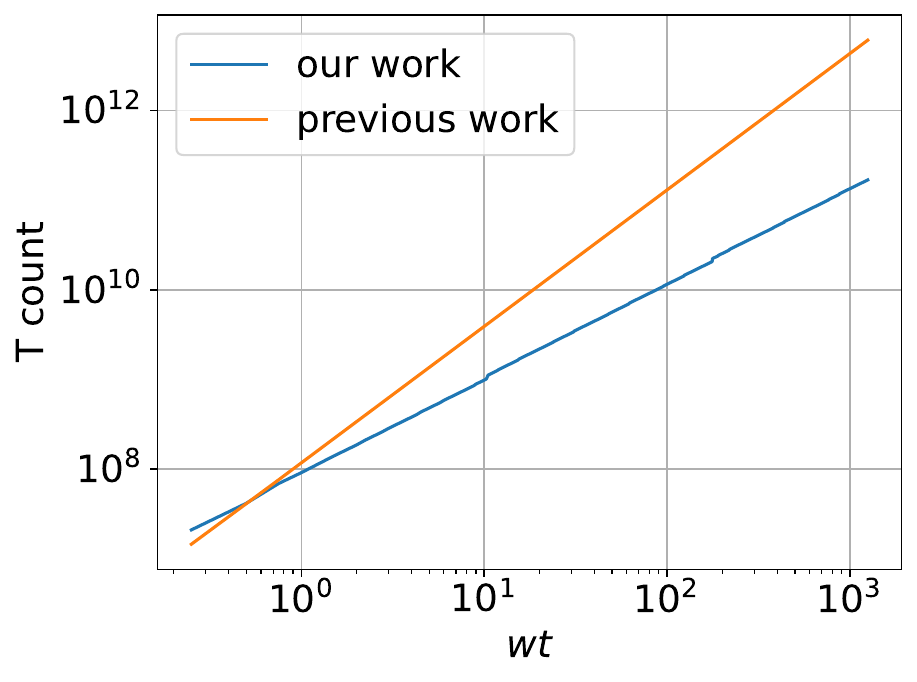}
    \label{fig:comp-T-time-dep}
}
\hfil 
\subfloat[Dependence on precision \( \varepsilon \). (\( N=64 \), \( wt=10 \))]{%
    \includegraphics[keepaspectratio, width=.3\linewidth]{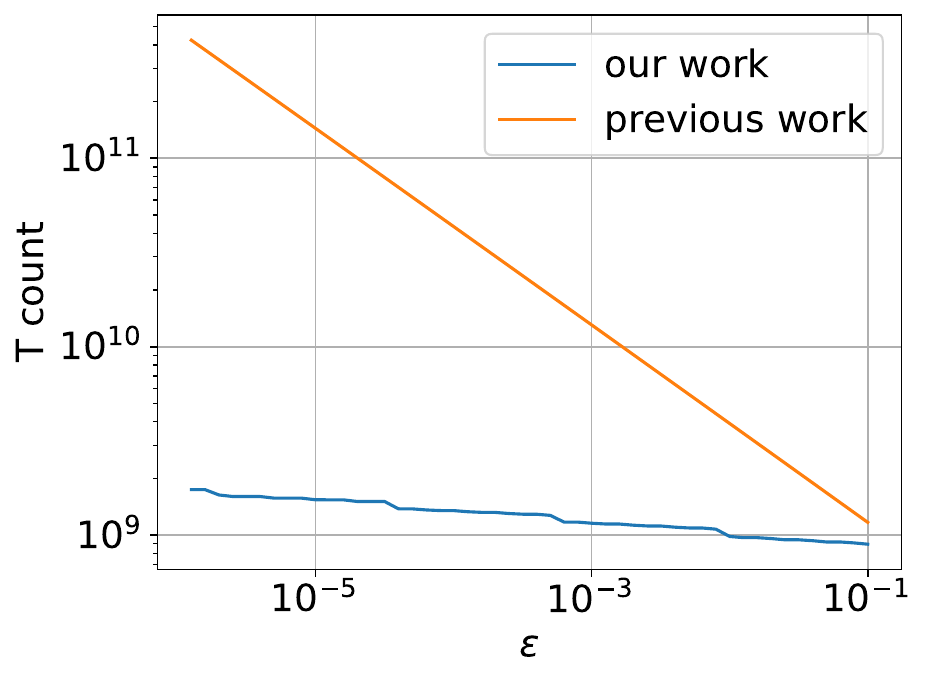}
    \label{fig:comp-T-precision-dep}
}

\caption{Comparison of the number of T gates for simulating the Schwinger model Hamiltonian. }
\label{fig:comp-T-overall}
\end{figure}

\subsection{Rough estimation of the number of physical qubits}
Using Result~\ref{cor:vacuum-amplitude-estimate-main}, we also roughly estimate the number of physical qubits assuming the surface code. 
Let us explain the method of estimating the number of physical qubits used in this work. 
Let $p_{\mathrm{phys}}$ and $p_{\mathrm{L}}$ be the physical error rate and the logical error rate achieved by error correction, respectively. 
For the standard scheme using the surface code with a code distance $d$, the logical error rate is approximately given by $p_L = 0.1(100 p_{\mathrm{phys}})^{(d+1)/2}$~\cite{litinski2019magic}. 
Let $M$ be the number of gates counted in terms of Clifford+T gate set. We roughly estimate $M$ to be 100 times the number of T gates.
This is justified by the fact that most of the subroutines used in the algorithm are primarily composed of Toffoli gates. 
For example, the number of Clifford gates used in inequality tests is at most four times the number of T-gates \mbox{\cite{berry2019qubitization}}.
Some subroutines are composed only of Clifford gates, but the number of such subroutines is almost the same as the number of subroutines mainly composed of Toffoli gates (see circuit descriptions of PREPARE operators in Appendix \mbox{\ref{appsubsec:prepare}}).
Furthermore, even if this estimate of the number of Clifford gates cannot be applied, the logical error rate decays exponentially with the code distance, so the influence on the number of physical qubits is minimal.
Next, we choose the code distance $d$ such that $p_L<1/M$ depending on $p_{\mathrm{phys}}$. 
Given the code distance $d$, the number of physical qubits required for one logical qubit is approximately $2d^2$. 
Let $N_{\mathrm{tot}}$ be the number of logical qubits required to run the algorithm (i.e. the number of qubits which we estimated in Result~\ref{cor:vacuum-amplitude-estimate-main}). 
In order to implement fault-tolerant quantum computation, we need ancillary logical qubits in addition to $N_{\mathrm{tot}}$ logical qubits for, e.g., routing and distillation. 
We assume having $4N_{\mathrm{tot}}$ in total is sufficient for those operations~\cite{litinski2019game}.
Thus, the total number of physical qubits is estimated as $4N_{\mathrm{tot}}\times 2d^2$. 

Fig.~\ref{fig:physical-qubit} shows the number of physical qubits assuming the physical error rate of $p_{\mathrm{phys}}=10^{-3}$ and $10^{-4}$ based on the assumptions described above and Result~\ref{cor:vacuum-amplitude-estimate-main}.  
Here, we set the parameters as $a=0.2$, $m=0.1$, $\theta =\pi$, $g=1$, $w=1/(2a) = 2.5$, and $J=g^2 a/2=0.1$, $wt=10$, $\varepsilon = 0.01$. 
For $N=64$ and $p_{\mathrm{phys}}=10^{-3}$ (or $10^{-4}$), we found that the number of physical qubits needed for calculating the vacuum persistence amplitude is about $9\times 10^{5}$ (or $2\times 10^{5}$), which is comparable to the estimate for Hubbard model estimated in Ref.~\cite{babbush2018encoding}.

\begin{figure}
\centerline{
\includegraphics[width=70mm, page=1]{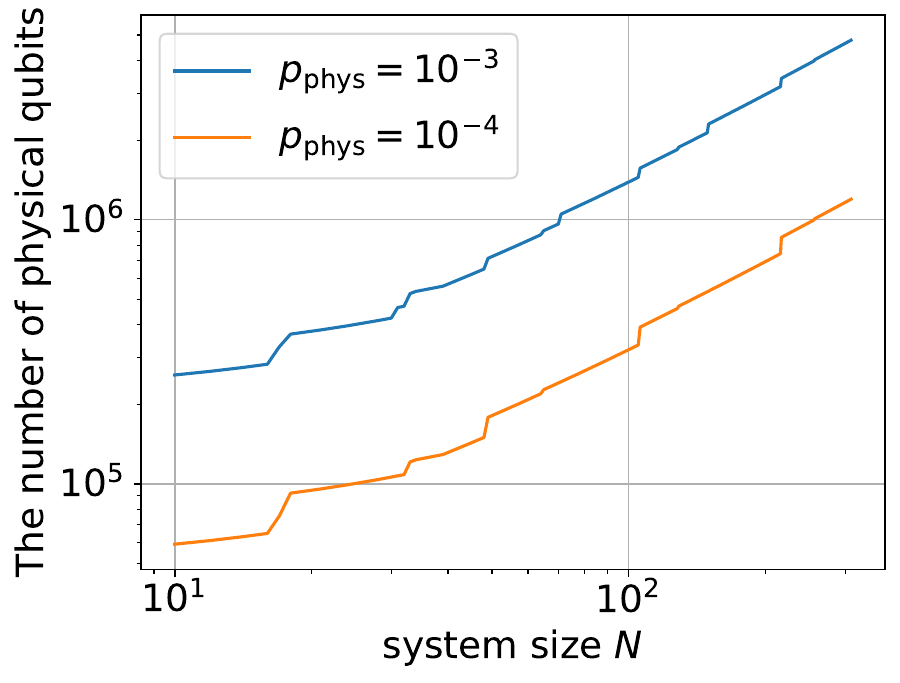}
}
\caption{The number of physical qubits required for estimating the vacuum persistence amplitude. }
\label{fig:physical-qubit}
\end{figure}

\section{Conclusion}\label{sec:conclusion}
The contributions of this article can be summarized in two main aspects. 
First, we have proposed an efficient implementation of block encoding of the Schwinger model Hamiltonian. 
We decompose the Hamiltonian into six terms based on coefficient differences and efficiently construct block-encodings of each term. By careful decomposition of the Hamiltonian, we can construct the block encoding without QROM. 
Furthermore, we believe that our techniques can be applied to more general Hamiltonians with all-to-all interactions and a constant number of parameters.
Second, we have provided an estimate for the resource required to solve the end-to-end problem, i.e., calculating the vacuum persistence amplitude. 
In the context of applications in quantum chemistry and condensed matter physics which persue ground state energies, it is often assumed that we have a state which has a large overlap with the ground state. 
On the other hand, when we estimate the vacuum persistence amplitude, there is no such assumption since N\'eel states can be generated using only X gates. 
This result provides an insight into whether fault-tolerant quantum computers can solve problems that are computationally challenging for classical computers within a realistic timeframe. 

We have also clarified several future challenges. Firstly, further improvements in quantum algorithms are required. 
As our results indicate, estimating the vacuum persistence amplitude with certain parameters is difficult. 
Therefore, we need to make efforts to improve a constant-factor performance of quantum algorithms and to reduce the normalization factor of the block-encoding. Additionally, the consumption rate of T gates may significantly limit the application of quantum computers. In the case of this article, the rate of 1 kHz is insufficient for solving the problem in a realistic timeframe, and the rate of 1 MHz is minimum requirement. 

Finally, let us discuss about the future directions of our results. One avenue is characterizing Hamiltonians that can reduce the number of SELECT operators in their block-encodings. In this work, we observed that three terms, $H_{Z\text{even}}$, $H_{Z\text{odd}}$, and $H_{Z^2}$, in the Schwinger model Hamiltonian have similar structures, which allowed us to minimize the use of SELECT operators. Conversely, it is worth establishing conditions for Hamiltonians that reduce the use of SELECT operators and provide block-encodings which require fewer T gates. Another avenue is seeking for models which have simple forms like the Schwinger model and are computationally difficult for classical computers. From the perspective of verifying the performance of quantum computer, estimating the cost of solving end-to-end problems in such models will provide valuable information for not only theoretical but also experimental research.

\begin{acknowledgments}
This work is supported by MEXT Q-LEAP Grant No. JPMXS0118067394 and JPMXS0120319794, JST COINEXT Grant No. JPMJPF2014, JST PRESTO Grant No. JPMJPR2019 and JPMJPR2113, JSPS KAKENHI Grant No. 23H03819 and 23H05439, JSPS Grant-in-Aid for Transformative Research Areas (A) JP21H05190, JST Grant Number JPMJPF2221 and JPMXP1020230411.
\end{acknowledgments}

%--------------------------------------------------

\bibliographystyle{quantum}
\bibliography{ref,Nf2_Schwinger,QFT}

%--------------------------------------------------

\appendix

\clearpage

\section{Subroutines}\label{appsec:subroutines}
In this section, we provide details about subroutines which we used to construct the block-encoding in Section. \ref{sec3}.
They are summarized as Table~\ref{tab:select},~\ref{tab:prepare}. 

\subsection{SELECT operator}\label{appsubsec:select}
We use multi-qubit controlled SELECT operators in our construction (see Fig.~\ref{fig:schwinger-block}).
These operations can be constructed with the cost listed in Table \ref{tab:select} using the following facts \cite{babbush2018encoding}.

\begin{itemize}
    \item Let $m\in\mathbb{N}$ and let $U$ be an unitary. Then a $m$-qubit controlled-$U$, that is, $\ket{0^m}\bra{0^m}\otimes U + (I-\ket{0^m}\bra{0^m})\otimes I$, can be implemented using $4m-4$ T gates, one single-qubit controlled-$U$, and $m-1$ ancilla qubits.
    \item Let $N\in\mathbb{N}, n\in \{0,\dots, N-1 \}$. Then a single-qubit controlled version of $\text{SELECT}$ operator, i.e.,  
\begin{equation}\label{eq:selectex}
\ket{1}\bra{1} \otimes \left( \sum_{n=0}^{N-1} \ket{n}\bra{n} \otimes U_n + \left( I - \sum_{n=0}^{N-1} \ket{n}\bra{n}\right) \otimes I\right) + \ket{0}\bra{0} \otimes I \otimes I, 
\end{equation}
can be implemented using $4N-4$ T gates, $\lceil\log_2{N}\rceil$ ancilla qubits, and one single-qubit controlled-$U_n$ for each $n$.  
\end{itemize}

 \begin{table}
    \centering
    \footnotesize
\caption{Subroutines for SELECT. The ancilla qubits are defined as those that can be reused in subsequent operations. $N$ represents system size of the Schwinger model Hamiltonian.}
\label{tab:select}
\renewcommand{\arraystretch}{2}
    \begin{tabular}{p{5cm}p{4cm}p{4cm}}\hline\hline
         Subroutine &  T Cost & Note\\\hline
 3-qubit-controlled $V_{XX}$ ($\text{c3-}V_{XX}$)& $4N$&Defined in \eqref{eqapp1.10}.\\
         3-qubit-controlled $V_{YY}$ ($\text{c3-}V_{YY}$)& $4N$&Defined in \eqref{eqapp1.11}.\\
         2-qubit-controlled $V_{Z}$ ($\text{c2-}V_{Z}$)& $4N$&Defined in \eqref{eqapp1.12}.\\
         3-qubit-controlled $V_{Z2}$ ($\text{c3-}V_{Z2}$)& $4N+4$&Defined in \eqref{eqapp1.16}.\\
         4-qubit-controlled $V_{Z2}$ ($\text{c4-}V_{Z2}$)& $4N+8$&Defined in \eqref{eqapp1.16}.\\ \hline\hline
    \end{tabular}
\end{table}

\subsection{PREPARE operators}\label{appsubsec:prepare}

\begin{table}
    \centering
    \scriptsize
\caption{Subroutines used in PREPARE. Unless otherwise specified, $\varepsilon$ denotes an error measured in operator norm. The ancilla qubits are defined as those that can be reused in subsequent operations.}
\label{tab:prepare}
    \begin{tabular}{p{1.5cm}p{0.7cm}p{3cm}p{2cm}p{1.45cm}p{3.8cm}}\hline\hline
         Subroutine &  Param &  T Cost & \# of ancilla & \# of unreusable ancilla & Note\\\hline
         $R_Y$ &  $\varepsilon$&  $4\log_2(1/\varepsilon)+C$ &0&0&Worst case result of \cite{ross2016optimal}. $C = 5+4 \log_2(1+\sqrt{2}) \approx 10.09$.\\
 $R_0$& $s$& $4s-8$& $s-2$&0&$s$-bit reflection operator, $R_0=2\ket{0^s}\bra{0^s}-I$.\\
 Inequality test ($\texttt{ineq}$)& $s$& $4s$& $s-1$&0&Calculates $a\leq b$ for $s$-bit integers in registers $\ket{a}$ and $\ket{b}$ \cite{gidney2018halving, berry2019qubitization}. Inversion can be done without T gate.\\
 controlled-$s$-bit-SWAP (\texttt{cswap})& $s$& $7s$& 0&0&\\
 Subtraction ($\texttt{sub}$)& $s$& $4s-4$& $s-1$&0&Calculates $a - b$ for $s$-bit integers $a$ and $b$ \cite{gidney2018halving}.\\
         $\mathrm{UNI}$& $N$, $\varepsilon$&  $2R_Y(\varepsilon/2)+2\texttt{ineq}(l)+R_0(l+1)$&$\texttt{ineq}(l)+l+1$ &1& $\eta$ and $L$ are defined as integers such that $N=2^\eta L$. $l=\lceil \log_2 L\rceil$. Additional $b$ qubits are for storing $\ket{L-1}$.\\
 controlled-UNI (cUNI)& $N$, $\varepsilon$& $\mathrm{UNI}+4\eta + 4l + 12$ & $\texttt{ineq}(l)+l+1$&1&\\
 $P_{S1}$& $N, \varepsilon$& $\mathrm{UNI}(N',\varepsilon/2)+\mathrm{UNI}(N'-1,\varepsilon/2)+\texttt{sub}(b')+\texttt{ineq}(b'+1)+\texttt{cswap}(b')+4$& \raggedright $\text{max} (\texttt{ineq}(b'+1)+\nu, \text{UNI}(N') + \text{UNI}(N'-1))$&$2b'+\tilde{b}'+2+\text{UNI}(N') + \text{UNI}(N'-1)$&$N'=\lceil N/2\rceil, b'=\lceil\log_2 N'\rceil, \tilde{b}'=\lceil\log_2 (N'-1)\rceil$, $\nu = b'+1 - \tilde{b}'$.\\
 controlled-$P_{S1}$ ($\text{c}P_{S1}$)& $N, \varepsilon$& $\mathrm{cUNI}(N',\varepsilon/2)+\mathrm{cUNI}(N'-1,\varepsilon/2)+2\texttt{sub}(b')+\texttt{ineq}(b'+1)+\text{controlled-}\texttt{cswap}(b')$& \raggedright $\text{max} (\texttt{ineq}(b'+1)+\nu+1, \text{cUNI}(N') + \text{cUNI}(N'-1)+1)$& $2b'+\tilde{b}'+2+\text{UNI}(N') + \text{UNI}(N'-1)$ &controlled-\texttt{cswap}$(b')$ needs $7b'+4$ T gates and one ancilla qubit.\\
 $P_{S2}$& $N, \varepsilon$& $2\mathrm{UNI}(N'',\varepsilon/2)+\texttt{sub}(b'')+\texttt{ineq}(b''+1)+\texttt{cswap}(b'')+4$& \raggedright $\text{max} (\texttt{ineq}(b''+1) +1, 2\text{UNI}(N''))$& $3b''+2+2\text{UNI}(N'')$ &$N''=\lfloor N/2\rfloor, b''=\lceil\log_2 N''\rceil$. \\
 controlled-$P_{S2}$ ($\text{c}P_{S2}$)& $N,\varepsilon$ & $2\mathrm{cUNI}(N'',\varepsilon/2)+2\texttt{sub}(b'')+\texttt{ineq}(b''+1)+\text{controlled-}\texttt{cswap}(b'')$& \raggedright $\text{max} (\texttt{ineq}(b''+1) +2, 2\text{cUNI}(N'')+1)$& $3b''+2+2\text{UNI}(N'')$ &$b=\lceil \log_2 N\rceil$.\\
 $P_{S3}'$& $N,\varepsilon$& $3\text{UNI}(N,\varepsilon/3)+\texttt{ineq}(b)+4$& $\text{max}(\text{UNI}(N,2\varepsilon),$\par$\texttt{ineq}(b))+1$& $b+1+\text{UNI}(N)$ & $\varepsilon$ represents the operator norm error of $R_Y$ appearing through the circuit.\\
 controlled-$P_{S3}'$ ($\text{c}P_{S3}'$)& $N,\varepsilon$& $\text{cUNI}(N,\varepsilon/3)+2\text{UNI}(N,\varepsilon/3)+\texttt{ineq}(b)+8$& $\text{max}(\text{cUNI}(N,2\varepsilon),$\par$\texttt{ineq}(b)+1)+1$& $b+1+\text{UNI}(N)$ &$b=\lceil \log_2 N\rceil$.\\
 $P_{S3}$& $N,\varepsilon$& $3P_{S3}'(N,6\varepsilon/20)+2R_Y(\varepsilon/20)+R_0(2b+3)+2R_0(b+3)$& $R_0(2b+3)$& $P_{S3}'(N)+2$ & $b=\lceil \log_2 N\rceil$.\\
 controlled-$P_{S3}$ ($\text{c}P_{S3}$)& $N,\varepsilon$& $\text{c}P_{S3}'(N,6\varepsilon/20)+2P_{S3}'(N,6\varepsilon/20)+2R_Y(\varepsilon/20)+R_0(2b+4)+2R_0(b+4)$& $R_0(2b+4)$& $P_{S3}'(N)+2$ &$b=\lceil \log_2 N\rceil$.\\
 $P_1$& $N, \varepsilon$& $\text{c}P_{S1}(N,\frac{4\varepsilon}{39})+\text{c}P_{S2}(N,\frac{4\varepsilon}{39})+\text{c}P_{S3}(N,\frac{20\varepsilon}{39})+\text{cUNI}(N-1,\frac{2\varepsilon}{39})+\text{cUNI}(N,\frac{2\varepsilon}{39})+7R_Y(\frac{\varepsilon}{39})+44$& $\text{max}(\text{c}P_{S1}(N,\varepsilon)+2,$ $\text{c}P_{S2}(N,\varepsilon)+2, \text{c}P_{S3}(N,\varepsilon)+1)$ & $\text{max}(\text{c}P_{S1}(N),$ $\text{c}P_{S2}(N),$ $ \text{c}P_{S3}(N))$ & The constant factor 44 is needed to make each of the operations multi-qubit controlled (see Fig.~\ref{fig:p1}). \\
 $\text{una}$ & $s$ & $4s-4$& $2s$&  & Produce a register which has zeros matching the leading zeros in $s$-bit integer $a$ and ones after that.\\
 $P_2$& $N,\varepsilon,\Delta$& $\frac{d-1}{2}(2R_Y(\varepsilon/d)+\texttt{ineq}(b)+8b+R_0(b+1))+R_Y(\varepsilon/d)+2\texttt{ineq}(b)+4b+\texttt{sub}(b)+\text{una}(b)$& $\texttt{ineq}(b)+\text{una}(b)+\texttt{sub}(b)+2b+2$& 0 & $d$ is the smallest odd integer larger than $\sqrt{2} \ln (2/\sqrt{\Delta})$. $b=\lceil \log_2 N\rceil$. The implemented $\tilde{P}_2$ satisfies $\|P_2'-\tilde{P}_2\|\leq \varepsilon$ where $P_2'$ is defined in Eq.~\eqref{eq:lemp2} and is an approximation to the desired $P_2$ defined in Eq.~\eqref{eqapp1.14}. $\Delta$ is an error parameter that measures the difference between $P_2'$ and the desired $P_2$ as $P_2'\ket{n}\otimes \ket{0}=\frac{\xi_n}{\sqrt{n}}\ket{n}\sum_{i=0}^{n-1}\ket{i} + \cdots$ and $1-|\xi_n|^2 > \Delta$.\\
 controlled-$P_2$ (c$P_2$)& $N,\varepsilon,\Delta$& $\frac{d-1}{2}(4R_Y(\varepsilon/(2d))+\texttt{ineq}(b)+8b+R_0(b+1))+2R_Y(\varepsilon/(2d))+2\texttt{ineq}(b)+8b+\texttt{sub}(b)+\text{una}(b)+4$& $\texttt{ineq}(b)+\text{una}(b)+\texttt{sub}(b)+2b+2$& 0 &\\\hline\hline
    \end{tabular}
\end{table}

In this subsection, we describe the constructions of $\mathrm{UNI}_{(N)}$, $P_{S1}$, $P_{S2}$, $P_{S3}$, and $P_2$. Results are summarized as Table~\ref{tab:prepare}.
We will use the following facts implicitly in this section.
\begin{itemize}
    \item Let $U_i$ and $V_i$ be unitaries for $i\in\{1,\dots M\}$, then
\begin{equation}\label{eq:unitary-subadditivity}
\left\| U_M \dots U_2 U_1 - V_M \dots V_2 V_1 \right\| \leq \sum_{i=1}^{M} \left\| U_i - V_i \right\|,
\end{equation}
as shown in \cite{nielsen2010quantum}. 
\item Let $R_Z(\theta)$ be a single-qubit Z rotation operator by angle $\theta$: 
\begin{equation}\label{eq:z-rotation}
R_Z(\theta)=e^{-i \theta Z / 2}=
\begin{bmatrix}
e^{-i \theta / 2} & 0 \\
0 & e^{i \theta / 2}
\end{bmatrix}.
\end{equation}
Then for any $\varepsilon>0$, we can find an operator $U$ expressible in the single-qubit Clifford+T gate set, such that 
\begin{equation}\label{eq:z-rotation2}
\left\| R_Z(\theta) - U \right\| \leq \varepsilon. 
\end{equation}
In the worst case, $U$ is implemented with $(4\log_2(1/\varepsilon)+C)$ T gates, where 
\begin{equation}\label{eq:C}
    C = 5+4 \log_2(1+\sqrt{2}) \approx 10.09
\end{equation} as shown in \cite{selinger2015efficient,ross2016optimal}. Note that X rotations and Y rotations can also be implemented with the same complexity since $X=HZH$ and $Y=SHZHS^\dag$. 
Moreover, a single-qubit controlled Y rotation can be implemented with twice the number of T gates since $\text{controlled-}R_Y(2\theta)=R_Y(\theta)\mathrm{CNOT} R_Y(-\theta)\mathrm{CNOT}$. 
\item The $s$-qubit reflection gate $R_0=2\ket{0^s}\bra{0^s}-I$ can be performed with $4s-8$ T gates and $s-2$ ancilla qubits.
\item Inequality test which calculates $a\leq b$ for $s$-bit integers in registers $\ket{a}$ and $\ket{b}$ can be performed with $4s$ T gates and $s-1$ ancilla qubits using the out-of-place adder \cite{gidney2018halving, berry2019qubitization}. Inversion can be done without T gate.
\item Subtraction which calculates $a - b$ for $s$-bit integers $a$ and $b$ can be performed with $4s-4$ T gates and $s-1$ ancilla qubits \cite{gidney2018halving}.
\end{itemize}

\subsubsection{\texorpdfstring{$\mathrm{UNI}_{(N)}$}{} operation}
We construct $\mathrm{UNI}_{(N)}$ following Ref. \cite{sanders2020compilation}. 
The steps closely follow that of \cite{sanders2020compilation}, but for completeness, we describe the procedure in detail.

\begin{lem}[Preparation of a uniform superposition state]\label{lem:uniform}
Let $N\in\mathbb{N}$ and let $b=\lceil \log_2 N \rceil$. Let $\eta$ and $L$ be integers satisfying $N=2^\eta \cdot L$, and $l=\lceil \log_2 L \rceil$. Define $\mathrm{UNI}_{(N)}$ as an operator satisfying,
\begin{equation}\label{eq:uniform}
\left(I\otimes \bra{1}R_Y(\theta) \right) \mathrm{UNI}_{(N)} \ket{0^b}\ket{0} = \frac{1}{\sqrt{N}}\sum_{n=0}^{N-1}\ket{n}, 
\end{equation}
where $\theta = 2\arcsin(2^{-1}\sqrt{2^l/L})$. 
For any $\varepsilon>0$, we can implement a unitary $\tilde{U}$ such that $\left\| \mathrm{UNI}_{(N)} - \tilde{U} \right\| \leq \varepsilon$
with $8\lceil\log_2{(2/\varepsilon)}\rceil + 12l + 2C- 4$ T gates and $2l$ ancilla qubits, where $C$ is the constant defined in Eq. \eqref{eq:C}. 
Moreover, a controlled version of $\tilde{U}$ can be implemented with $8\lceil\log_2{(2/\varepsilon)}\rceil +4\eta + 16l +  2C + 8$ T gates and $2l$ ancilla qubits. 
\end{lem}
\begin{proof}
First, we apply Hadamard gates on the $\eta$ qubits since
\begin{equation}\label{eq:lemuni:separation}
\frac{1}{\sqrt{N}}\sum_{n=0}^{N-1}\ket{n} = \frac{1}{\sqrt{L}}\sum_{i=0}^{L-1}\ket{i} \frac{1}{\sqrt{2^\eta}}\sum_{j=0}^{2^\eta-1}\ket{j}.
\end{equation}
We thus consider how to generate the state $\frac{1}{\sqrt{L}}\sum_{i=0}^{L-1}\ket{i}$ below.
\begin{enumerate}
\item Perform Hadamard gates on $l$ qubits. 
\item Add one ancilla qubit and apply $R_Y$, to make an amplitude for success 1/2, i.e.,
\begin{equation}\label{eq:lemuni:rotation}
    \left(\sqrt{\frac{L}{2^l}}\frac{1}{\sqrt{L}}\sum_{i=0}^{L-1}\ket{i}+ \sqrt{\frac{2^l-L}{2^l}}\frac{1}{\sqrt{2^l-L}}\sum_{i=L}^{2^l-1}\ket{i}\right)\ket{0}
    \rightarrow \frac{1}{2}\frac{1}{\sqrt{L}}\sum_{i=0}^{L-1}\ket{i}\ket{1} + \ket{\perp},
\end{equation}
where $\ket{\perp}$ is an unnormalized state which is orthogonal to $\frac{1}{2}\frac{1}{\sqrt{L}}\sum_{i=0}^{L-1}\ket{i}\ket{1}$.
\item Reflect about the target state $\frac{1}{\sqrt{L}}\sum_{i=0}^{L-1}\ket{i}\ket{1}$ using the inequality tests and a controlled Z gate. 
\item Reflect about the state in Eq. \eqref{eq:lemuni:rotation} using Hadamard gates, the inverse of the Y rotation, and $R_0$. 
\item Perform the inequality tests again to flag success of the state preparation. 
\end{enumerate}
We show the implementation of $U$ in Fig. \ref{fig:uniform}.
Considering the error of Y rotations, we can implement $\tilde{U}$ performing the Y rotations with $\varepsilon/2$-precision (see Eq.~\eqref{eq:unitary-subadditivity}). 
As a result, the overall cost of this procedure is $2(4\lceil\log_2{(2/\varepsilon)}\rceil+C)+12l-4$ T gates and $\max{(2l-1,l-1)}+1=2l$ ancilla qubits. 
Furthermore, we can implement a controlled version of $\tilde{U}$ replacing Hadamard gates in the first step, the controlled Z gate, $R_0$, and the CNOT gate with their controlled versions, respectively. 
Thus we need the additional cost of $4\eta + 4l + 12$ T gates since a controlled Hadamard gate requires $4$ T gates \cite{lee2021even}. 
\end{proof}
\begin{figure}[htbp]
\centerline{
\includegraphics[width=140mm, page=1]{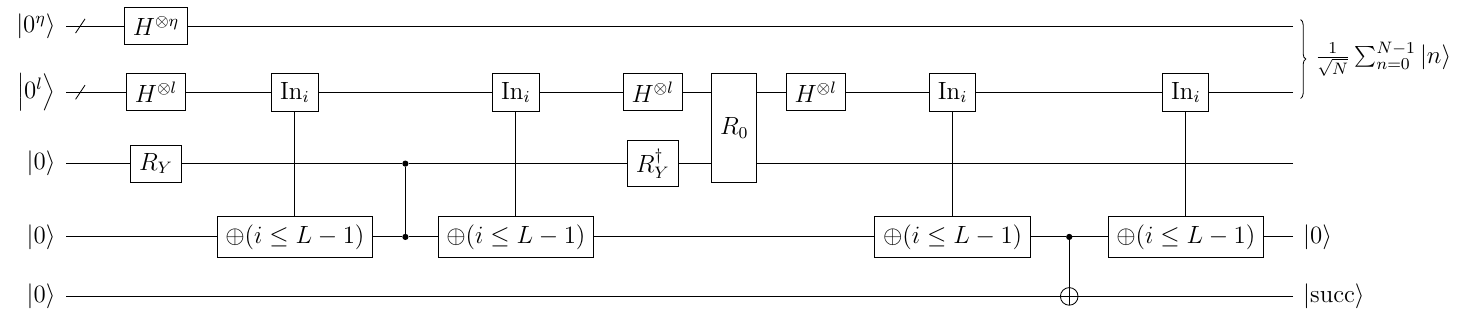}
}
\caption{Circuit $\text{UNI}_{(N)}$ for preparing a uniform superposition state.}
\label{fig:uniform}
\end{figure}

\subsubsection{\texorpdfstring{$P_{S1}$}{} and \texorpdfstring{$P_{S2}$}{} operation}
We construct $P_{S1}$ as,
\begin{equation}\label{eq:ps1:variable}
P_{S1}\ket{0^b}\ket{0^{\beta_{S_1}}} = \frac{1}{\sqrt{\alpha_{S1}}} \sum_{\substack{n=1 \\ n: \mathrm{even}}}^{N-1} \sqrt{n}\ket{n} \ket{\mathrm{temp}_n} , 
\end{equation}
instead of the one defined naively in \eqref{eqapp1.13}.
$\ket{\mathrm{temp}_n}$ is a $\beta_{S_1}=2\lceil\log_2 N' \rceil + \lceil\log_2 (N'-1) \rceil +4$-qubit unspecified junk register entangled with $\ket{n}$, where the number of unreusable ancilla qubits used in $\mathrm{UNI}$ operator is contained in $\beta_{S_1}$. 
The state preparation procedures introduced in Lemma. \ref{lem:ps1} and Lemma. \ref{lem:ps2} is based on the method proposed in Ref.~\cite{babbush2018encoding}.
Using the notations in \cite{babbush2018encoding}, for preparing $\sum_{n=1}^{N'-1} \sqrt{n}\ket{n} \ket{\mathrm{temp}_n}$, we set $\texttt{alt}_n = N'-1-n$  and $\texttt{keep}_n = 2n$ to achieve our goal. In the following, we describe a more detailed procedure.
\begin{lem}[$P_{S1}$ operator]\label{lem:ps1}
Let $N\in\mathbb{N}$, $b=\lceil \log_2 N \rceil$, $C$ be a constant defined in \eqref{eq:C}, $N'=\lceil N/2 \rceil$, $\eta', \mu, L',K'$ be integers such that $N' = 2^{\eta'}L', \ N'-1 = 2^{\mu'}K'$. 
For any $\varepsilon>0$, we can implement a unitary $\tilde{P}_{S1}$ that approximates $P_{S1}$ in \eqref{eq:ps1:variable} as $\left\| P_{S1} - \tilde{P}_{S1} \right\| \leq \varepsilon$, with $16\lceil\log_2{(4/\varepsilon)}\rceil +15\lceil\log_2 N' \rceil + 12 \lceil\log_2 L' \rceil + 12 \lceil\log_2 K' \rceil +  4C-4$ T gates and $\max (2\lceil\log_2 N' \rceil - \lceil\log_2 (N'-1) \rceil +1, 2\lceil\log_2 L'\rceil + 2\lceil\log_2 K' \rceil) $ ancilla qubits. 
Moreover, a controlled version of $\tilde{P}_{S1}$ can be implemented with $16\lceil\log_2{(4/\varepsilon)}\rceil +19\lceil\log_2 N' \rceil + 4\eta' + 4\mu' + 16 \lceil\log_2 L' \rceil + 16 \lceil\log_2 K' \rceil +  4C + 16$ T gates and $\max (2\lceil\log_2 N' \rceil - \lceil\log_2 (N'-1) \rceil +2, 2\lceil\log_2 L'\rceil + 2\lceil\log_2 K' \rceil+1)$ ancilla qubits. 
\end{lem}
\begin{proof}
First, we rewrite the target state as follows:
\begin{equation}\label{eq:lemps1:rewrite}
\frac{1}{\sqrt{\alpha_{S1}}} \sum_{\substack{n=1 \\ n: \mathrm{even}}}^{N-1} \sqrt{n}\ket{n}
= \frac{1}{\sqrt{\alpha_{S1}}} \sum_{n=1}^{N'-1} \sqrt{2n}\ket{n} \ket{0} 
= \sqrt{\frac{2}{\alpha_{S1}}} \sum_{n=1}^{N'-1} \sqrt{n}\ket{n} \ket{0}.
\end{equation}
Thus we only have to consider how to prepare the state $\sqrt{\frac{2}{\alpha_{S1}}} \sum_{n=1}^{N'-1} \sqrt{n}\ket{n}$ and add one ancilla qubit $\ket{0}$. 
Using the above fact, $P_{S1}$ can be realized by the circuit in Fig. \ref{fig:ps1}. 
We describe each step of the circuit:
\begin{enumerate}
\item Prepare uniform superposition states on the two registers: 
\begin{equation}\label{eq:lemps1:step1}
\frac{1}{\sqrt{N'}}\sum_{i=0}^{N'-1}\ket{i} \frac{1}{\sqrt{N'-1}} \sum_{j=0}^{N'-2}\ket{j} \ket{0} \ket{0} \ket{0}.
\end{equation}
We flag simultaneous success of the procedure using a one Toffoli gate. 
\item Calculate $N'-1-i$. The state after this step is:
\begin{equation}\label{eq:lemps1:step2}
\frac{1}{\sqrt{N'}}\sum_{i=0}^{N'-1}\ket{i} \frac{1}{\sqrt{N'-1}} \sum_{j=0}^{N'-2}\ket{j} \ket{N'-1-i} \ket{0} \ket{0}.
\end{equation}
\item Calculate $2i$. It only requires CNOT gates. After this operation, the state becomes:
\begin{equation}\label{eq:lemps1:step3}
\frac{1}{\sqrt{N'}}\sum_{i=0}^{N'-1}\ket{i} \frac{1}{\sqrt{N'-1}} \sum_{j=0}^{N'-2}\ket{j} \ket{N'-1-i} \ket{2i} \ket{0}.
\end{equation}
\item Perform an inequality test:
\begin{equation}
\frac{1}{\sqrt{N'(N'-1)}} \sum_{i=0}^{N'-1} \left( \ket{i}\sum_{j=0}^{\min{(2i-1,N'-2)}}\ket{j} \ket{N'-1-i} \ket{2i} \ket{0}
+\ket{i}\sum_{j=\min{(N'-1,2i)}}^{N'-2}\ket{j} \ket{N'-1-i} \ket{2i} \ket{1}  \right).
\end{equation}
 
\item Perform controlled SWAP gates with $7\lceil\log_2 N' \rceil$ T gates \cite{amy2013meet}:
\begin{equation}\label{eq:lemps1:step4}
\frac{1}{\sqrt{N'(N'-1)}} \sum_{i=0}^{N'-1} \left( \ket{i}\sum_{j=0}^{\min{(2i-1,N'-2)}}\ket{j} \ket{N'-1-i} \ket{2i} \ket{0}
+\ket{N'-1-i}\sum_{j=\min{(N'-1,2i)}}^{N'-2}\ket{j} \ket{i} \ket{2i} \ket{1}  \right).
\end{equation}
\end{enumerate}
The state in Eq. \eqref{eq:lemps1:step4} is equivalent to $\sqrt{\frac{2}{\alpha_{S1}}} \sum_{n=1}^{N'-1} \sqrt{n}\ket{n}\ket{\text{temp}_n}$.
This can be seen as follows.
Projecting the state onto the subspace spanned by $\ket{n}\otimes I\otimes I\otimes I \otimes I$, we obtain
\begin{equation}\label{eq:lemps1:step4-projected}
\begin{split}
\frac{1}{\sqrt{N'(N'-1)}}& \left( \ket{n}\sum_{j=0}^{\min{(2n-1,N'-2)}}\ket{j} \ket{N'-1-n} \ket{2n} \ket{0} \right. \\
& \left. +\ket{n}\sum_{j=\min{(N'-1,2(N'-1-n))}}^{N'-2}\ket{j} \ket{N'-1-n} \ket{2(N'-1-n)} \ket{1}  \right).
\end{split}
\end{equation}
When $n<\frac{N'-1}{2}$, i.e.,$2n-1<N'-2$ and $2(N'-1-n)>N'-1$, the norm of this vector is
\begin{equation}\label{eq:lemps1:step4-norm}
\frac{1}{\sqrt{N'(N'-1)}}\sqrt{2n+0}= \sqrt{\frac{2}{N'(N'-1)}}\sqrt{n} \propto \sqrt{n}.
\end{equation}
When $n\geq\frac{N'-1}{2}$, i.e.,$2n-1\geq N'-2$ and $2(N'-1-n)\leq N'-1$, the norm of this vector is
\begin{equation}\label{eq:lemps1:step4-norm-2}
\frac{1}{\sqrt{N'(N'-1)}}\sqrt{N'-1+N'-2-(2(N'-1-n)-1)}= \sqrt{\frac{2}{N'(N'-1)}}\sqrt{n} \propto \sqrt{n}.
\end{equation}
Thus the circuit in Fig. \ref{fig:ps1} generates the state $\frac{1}{\sqrt{\alpha_{S1}}} \sum_{\substack{n=1 \\ n: \mathrm{even}}}^{N-1} \sqrt{n}\ket{n} \ket{\mathrm{temp}_n}$.

$\tilde{P}_{S1}$ uses $\text{UNI}_{(N')}$, $\text{UNI}_{(N'-1)}$, subtraction of $\lceil \log_2 N'\rceil$-bit integers, inequality test of $\lceil \log_2 N'\rceil$-bit integers, controlled-SWAP, and a Toffoli gate for flagging. Therefore, the overall cost of $\tilde{P}_{S1}$ is $16\lceil\log_2{(4/\varepsilon)}\rceil +15\lceil\log_2 N' \rceil + 12 \lceil\log_2 L' \rceil + 12 \lceil\log_2 K' \rceil +  4C-4$ T gates and $\max (2\lceil\log_2 N' \rceil - \lceil\log_2 (N'-1) \rceil + 1 , 2\lceil\log_2 L'\rceil + 2\lceil\log_2 K' \rceil) $ ancilla qubits. 
On the other hand, a controlled version of $P_{S1}$ can be implemented as Fig. \ref{fig:ps1-controlled}.
The cost of each step is as follows:
\begin{itemize}
\item Preparing uniform superpositions with $16\lceil\log_2{(4/\varepsilon)}\rceil + 4\eta' + 4\mu' + 16 \lceil\log_2 L' \rceil + 16 \lceil\log_2 K' \rceil +  4C + 16$ T gates and $2\lceil\log_2 L'\rceil + 2\lceil\log_2 K'\rceil +1$ ancilla qubits.
\item Subtraction with $4\lceil\log_2 N' \rceil-4$ T gates and $\lceil\log_2 N' \rceil - 1$ ancilla qubits. 
\item Inequality test and its inverse using out-of-place adder with $4\lceil\log_2 N' \rceil + 4$ T gates and $2\lceil\log_2 N' \rceil - \lceil\log_2 (N'-1) \rceil + 1$ ancilla qubits.
\item Controlled SWAP gates with $7\lceil\log_2 N' \rceil+4$ T gates and one ancilla qubit. 
\item Inverse of subtraction with $4\lceil\log_2 N' \rceil-4$ T gates and $\lceil\log_2 N' \rceil - 1$ ancilla qubits.
\end{itemize}
Thus overall cost of controlled $\tilde{P}_{S1}$ is $16\lceil\log_2{(4/\varepsilon)}\rceil +19\lceil\log_2 N' \rceil + 4\eta' + 4\mu' + 16 \lceil\log_2 L' \rceil + 16 \lceil\log_2 K' \rceil +  4C + 16$ T gates and $\max (2\lceil\log_2 N' \rceil - \lceil\log_2 (N'-1) \rceil + 2 , 2\lceil\log_2 L'\rceil + 2\lceil\log_2 K' \rceil+1)$ ancilla qubits.

\end{proof}
\begin{figure}
\centerline{
\includegraphics[width=125mm, page=1]{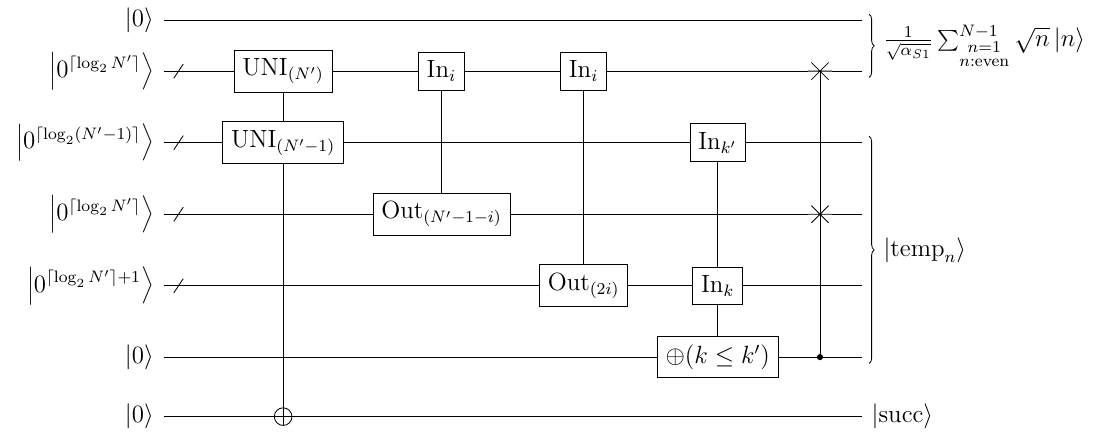}
}
\caption{Circuit for $P_{S1}$ operator.}
\label{fig:ps1}
\end{figure}

\begin{figure}
\centerline{
\includegraphics[width=150mm, page=1]{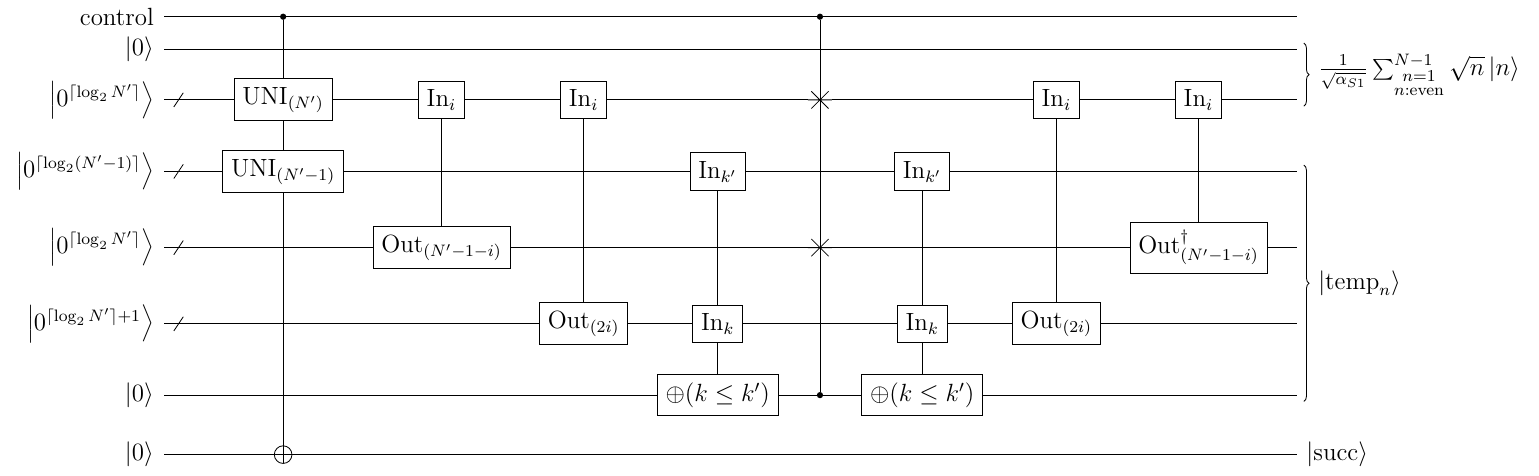}
}
\caption{Circuit for controlled $P_{S1}$ operator.}
\label{fig:ps1-controlled}
\end{figure}

$P_{S2}$ operator can be implemented in a similar way. 
\begin{lem}[$P_{S2}$ operator]\label{lem:ps2}
Let $N\in\mathbb{N}$, $b=\lceil \log_2 N \rceil$, $C$ be a constant defined in \eqref{eq:C}, $N''=\lfloor N/2 \rfloor$, $N'' = 2^{\eta''} \cdot L'' (L'': \text{odd})$. 
Define a unitary $P_{S2}$ that satisfies: 
\begin{equation}\label{eq:lemps2-2}
P_{S2}\ket{0^b}\ket{0^{\beta_{S_2}}} = \frac{1}{\sqrt{\alpha_{S2}}} \sum_{\substack{n=1 \\ n: \mathrm{odd}}}^{N-1} \sqrt{n}\ket{n} \ket{\mathrm{temp}_n} , 
\end{equation}
where $\alpha_{S2}$ is a normalization factor and $\ket{\mathrm{temp}_n}$ is a $\beta_{S_2}=3\lceil\log_2 N'' \rceil +4$-qubit unspecified junk register entangled with $\ket{n}$. 
For any $\varepsilon>0$, we can implement a unitary $\tilde{P}_{S2}$ such that
$\left\| P_{S2} - \tilde{P}_{S2} \right\| \leq \varepsilon$
with $16\lceil\log_2{(4/\varepsilon)}\rceil +15\lceil\log_2 N'' \rceil + 24 \lceil\log_2 L'' \rceil +  4C-4$ T gates and $\max (\lceil\log_2 N'' \rceil +1, 4\lceil\log_2 L''\rceil) $ ancilla qubits. 
Moreover, a controlled version of $\tilde{P}_{S2}$ can be implemented with $16\lceil\log_2{(4/\varepsilon)}\rceil +19\lceil\log_2 N'' \rceil + 8\eta'' + 32 \lceil\log_2 L'' \rceil +  4C + 16$ T gates and $\max (\lceil\log_2 N'' \rceil +1, 4\lceil\log_2 L''\rceil)+1$ ancilla qubits. 
\end{lem}
\begin{proof}
First, we rewrite the target state as below: 
\begin{equation}\label{eq:lemps2:rewrite}
\frac{1}{\sqrt{\alpha_{S2}}} \sum_{\substack{n=1 \\ n: \mathrm{odd}}}^{N-1} \sqrt{n}\ket{n}
= \frac{1}{\sqrt{\alpha_{S2}}} \sum_{n=0}^{N''-1} \sqrt{2n+1}\ket{n} \ket{1} 
= \sqrt{\frac{2}{\alpha_{S2}}} \sum_{n=0}^{N''-1} \sqrt{n+\frac{1}{2}}\ket{n} \ket{1} ,
\end{equation}
where $N''=\lfloor N/2\rfloor$. 
Thus we only have to prepare the state $\sqrt{\frac{2}{\alpha_{S2}}} \sum_{n=0}^{N''-1} \sqrt{n+\frac{1}{2}}\ket{n}$ and add one ancilla qubit $\ket{1}$. 
We provide a quantum circuit for $P_{S2}$ in Fig. \ref{fig:ps2}.
It is the same as the one for $P_{S1}$ except that the initial uniform superposition is both over $N''$ computational basis and that in atep 3 we calculate $2i+1$ instead of $2i$.
As such, we omit the detailed description of each circuit.
The final state after controlled SWAP gates reads:
\begin{equation}\label{eq:lemps2:step4}
\frac{1}{N''} \sum_{i=0}^{N''-1} \left( \ket{i}\sum_{j=0}^{\min{(2i,N''-1)}}\ket{j} \ket{N''-1-i} \ket{2i+1} \ket{0}
+\ket{N''-1-i}\sum_{j=\min{(N'',2i+1)}}^{N''-1}\ket{j} \ket{i} \ket{2i+1} \ket{1}  \right).
\end{equation}

$\tilde{P}_{S2}$ uses $\text{UNI}_{(N'')}$ two times, subtraction of $\lceil \log_2 N''\rceil$-bit integers, inequality test of $\lceil \log_2 N''\rceil$-bit integers, and controlled-SWAP.
The overall cost of $\tilde{P}_{S2}$ is $16\lceil\log_2{(4/\varepsilon)}\rceil +15\lceil\log_2 N'' \rceil + 24 \lceil\log_2 L'' \rceil +  4C-4$ T gates and $\max (\lceil\log_2 N'' \rceil + 1 , 4\lceil\log_2 L''\rceil) $ ancilla qubits. 

A controlled version of $P_{S2}$ can be implemented as in Fig. \ref{fig:ps2-controlled}.
The cost of each step is as follows:
\begin{itemize}
\item Preparing uniform superpositions with $16\lceil\log_2{(4/\varepsilon)}\rceil + 8\eta'' + 32 \lceil\log_2 L'' \rceil + 4C + 16$ T gates and $4\lceil\log_2 L''\rceil +1$ ancilla qubits.
\item Subtraction with $4\lceil\log_2 N'' \rceil-4$ T gates and $\lceil\log_2 N'' \rceil - 1$ ancilla qubits. 
\item Inequality test and its inverse using out-of-place adder with $4\lceil\log_2 N'' \rceil+4$ T gates and $\lceil\log_2 N'' \rceil + 1$ ancilla qubits.
\item Controlled SWAP gates with $7\lceil\log_2 N'' \rceil+4$ T gates and one ancilla qubit. 
\item Inverse of subtraction with $4\lceil\log_2 N'' \rceil-4$ T gates and $\lceil\log_2 N'' \rceil - 1$ ancilla qubits.
\end{itemize}
Thus overall cost of controlled $\tilde{P}_{S2}$ is $16\lceil\log_2{(4/\varepsilon)}\rceil +19\lceil\log_2 N'' \rceil + 8\eta'' + 32 \lceil\log_2 L'' \rceil +  4C + 16$ T gates and $\max (\lceil\log_2 N'' \rceil + 2 , 4\lceil\log_2 L''\rceil +1)$ ancilla qubits. 

Lastly, we confirm that the state in Eq. \eqref{eq:lemps2:step4} is equivalent to $\sqrt{\frac{2}{\alpha_{S2}}} \sum_{n=0}^{N''-1} \sqrt{n+\frac{1}{2}}\ket{n}$. 
Projecting the state onto the subspace spanned by $\ket{n}\otimes I\otimes I\otimes I \otimes I$, we obtain
\begin{equation}\label{eq:lemps2:step4-projected}
\begin{split}
\frac{1}{N''} \sum_{i=0}^{N''-1}& \left( \ket{n}\sum_{j=0}^{\min{(2n,N''-1)}}\ket{j} \ket{N''-1-n} \ket{2n+1} \ket{0} \right.\\
& \left. +\ket{n}\sum_{j=\min{(N'',2(N''-1-n)+1)}}^{N''-1}\ket{j} \ket{N''-1-n} \ket{2(N''-1-n)+1} \ket{1}  \right).
\end{split}
\end{equation}
When $n<\frac{N''-1}{2}$, i.e.,$2n<N''-1$ and $2(N''-1-n)+1>N''$, the norm of this vector is
\begin{equation}\label{eq:lemps2:step4-norm}
\frac{1}{N''}\sqrt{2n+1}= \frac{\sqrt{2}}{N''}\sqrt{n+\frac{1}{2}} \propto \sqrt{n+\frac{1}{2}}.
\end{equation}
When $n\geq\frac{N''-1}{2}$, i.e.,$2n\geq N''-1$ and $2(N''-1-n)+1\leq N''$, the norm of this vector is
\begin{equation}\label{eq:lemps2:step4-norm-2}
\frac{1}{N''}\sqrt{N''+N''-1-(2(N''-1-n)+1-1)}= \frac{\sqrt{2}}{N''}\sqrt{n+\frac{1}{2}} \propto \sqrt{n+\frac{1}{2}}.
\end{equation}
Thus the circuit in Fig. \ref{fig:ps2} generates the state $\frac{1}{\sqrt{\alpha_{S2}}} \sum_{\substack{n=1 \\ n: \mathrm{odd}}}^{N-1} \sqrt{n}\ket{n} \ket{\mathrm{temp}_n}$. 
\end{proof}
\begin{figure}
\centerline{
\includegraphics[width=125mm, page=1]{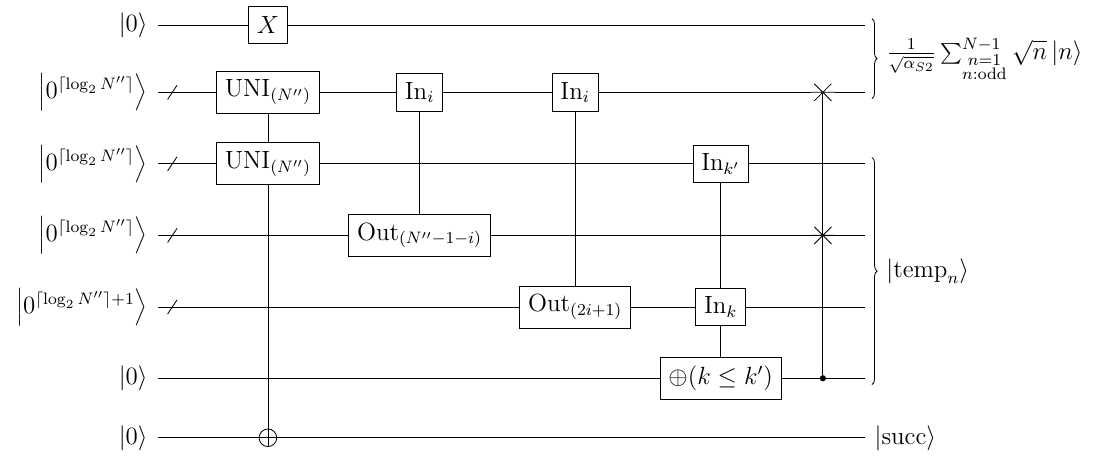}
}
\caption{Circuit for $P_{S2}$ operator.}
\label{fig:ps2}
\end{figure}
\begin{figure}
\centerline{
\includegraphics[width=150mm, page=1]{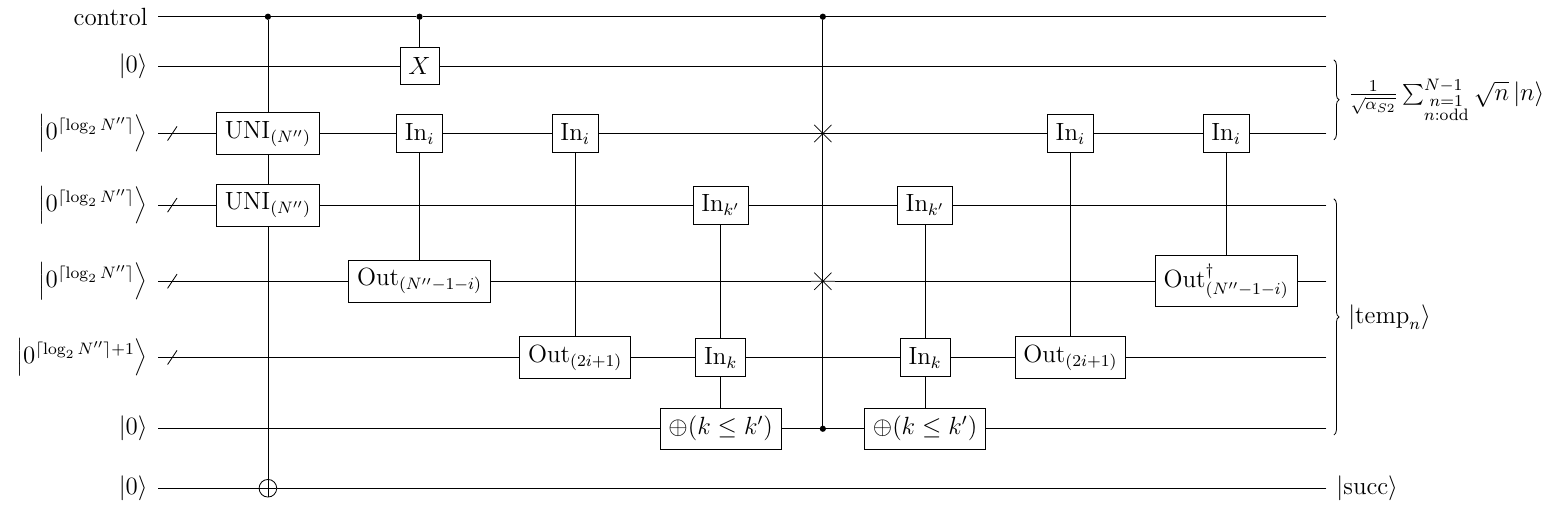}
}
\caption{Circuit for controlled $P_{S2}$ operator.}
\label{fig:ps2-controlled}
\end{figure}

\subsubsection{\texorpdfstring{$P_{S3}$}{} operator}

\begin{lem}[$P_{S3}$ operator]\label{lem:ps3}
Let $N\geq 8 \in\mathbb{N}$, $b=\lceil \log_2 N \rceil$, let $C$ be a constant defined in \eqref{eq:C}, $\eta,L$ be integers such that $N=2^\eta \cdot L$, and $l=\lceil \log_2 L \rceil$. 
Define $P_{S3}$ as a unitary satisfying: 
\begin{equation}\label{eq:lemps3-1}
(I\otimes \bra{0^{b}} \bra{1} \bra{1} \bra{1}R_Y(\theta)) P_{S3} \ket{0^b}\ket{0^{b}}\ket{0}\ket{0}\ket{0} = \frac{1}{\sqrt{\alpha_{S3}}} \sum_{n=1}^{N-1} n\ket{n}, 
\end{equation}
where $\alpha_{S3}$ is a normalization factor and $\theta = 2 \arcsin(\sqrt{3N^3/(2N(N-1)(2N-1))})$. 
For any $\varepsilon>0$, we can implement a unitary $\tilde{P}_{S3}$ such that
$\left\| P_{S3} - \tilde{P}_{S3} \right\| \leq \varepsilon$
with $80 \lceil\log_2{(20/\varepsilon)}\rceil + 28 b + 108 l + 20C - 12$ T gates, one unreusable ancilla qubit, and $2b$ ancilla qubits. 
Moreover, a controlled version of $\tilde{P}_{S3}$ can be implemented with $80 \lceil\log_2{(20/\varepsilon)}\rceil + 28 b + 4\eta + 112 l + 20C + 12$ T gates, one unreusable ancilla qubit, and $2b+1$ ancilla qubits. 
\end{lem}
\begin{proof}
We construct $P_{S3}$ based on the technique proposed in Ref. \cite{sanders2019black}. 
First, we define $P_{S3}'$ operator as in Fig. \ref{fig:ps3prime} which provides the target state in a specific subspace. 
This circuit works as follows: 
\begin{enumerate}
\item Prepare uniform superposition states on the first two register:
\begin{equation}\label{eq:lemps3:prime-step1}
\frac{1}{\sqrt{N}}\sum_{n=0}^{N-1}\ket{n}\frac{1}{\sqrt{N}}\sum_{n'=0}^{N-1}\ket{n'}  \ket{0} \ket{1},
\end{equation}
Note that it needs two unreusable ancilla qubits.
\item Perform an inequality test: 
\begin{equation}\label{eq:lemps3:prime-step2}
\frac{1}{N}\sum_{n=0}^{N-1}\ket{n} \left( \sum_{n'=0}^{n-1}\ket{n'}  \ket{1} + \sum_{n'=n}^{N-1}\ket{n'}  \ket{0} \right) \ket{1}.
\end{equation}
\item Invert the uniform superposition procedure. Now, one of the two unreusable ancilla qubits becomes reusable. 
\end{enumerate}
Projecting the state after this process onto the subspace spanned by $I\otimes \ket{0^b}\ket{1}\ket{1}$, we obtain the target state, i.e.,
\begin{align}\label{eq:lemps3:prime-projected}
&\left(I\otimes \bra{0^b}\bra{1}\bra{1}\right) \left(I\otimes \text{UNI}_{(N)} \otimes I \otimes I\right)
\frac{1}{N}\sum_{n=0}^{N-1}\ket{n} \left( \sum_{n'=0}^{n-1}\ket{n'}  \ket{1} + \sum_{n'=n}^{N-1}\ket{n'}  \ket{0} \right) \ket{1} \\
&=\left(I\otimes \frac{1}{\sqrt{N}}\sum_{i=0}^{N-1}\bra{i}\bra{1}\bra{1}\right)
\frac{1}{N}\sum_{n=0}^{N-1}\ket{n} \left( \sum_{n'=0}^{n-1}\ket{n'}  \ket{1} + \sum_{n'=n}^{N-1}\ket{n'}  \ket{0} \right) \ket{1} \\
&= \frac{1}{N\sqrt{N}}\sum_{n=1}^{N-1} n \ket{n}.
\end{align}
For $N\geq8$, the norm of this vector $\sqrt{\frac{(N-1)N(2N-1)}{6N^3}}$ is larger than $1/2$.
Thus, after adding a qubit to make the amplitude exactly $1/2$, we can perform a single iteration of amplitude amplification to obtain the desired state. This leads to $P_{S3}$ operator as in Fig. \ref{fig:ps3}. 

Finally, we summarize the cost of each operator. 
Noting that there are 20 $R_Y$ rotations in Fig. \ref{fig:ps3}, setting the error of each rotation to $\varepsilon/20$ give $\tilde{P}_{S3}$ with $\varepsilon$-precision.
Using this precision for rotations, the cost of $P_{S3}'$ is $24\lceil\log_2{(20/\varepsilon)}\rceil + 4b + 36l + 6C- 8$ T gates, one unreusable ancilla qubit, and $\max{(2l,b-1)}+2$ ancilla qubits.
Note that we can use one of the ancilla qubits of $\mathrm{UNI}_{(N)}$ as the ancilla qubit for the inequality test. 
A controlled version of ${P}_{S3}'$ can be implemented by replacing $\mathrm{UNI}_{(N)}$ on the first register and a Toffoli gate with controlled versions of them respectively. 
Therefore, controlled ${P}_{S3}'$ requires $24\lceil\log_2{(20/\varepsilon)}\rceil + 4b + 4\eta + 40l + 6C+4$ T gates and $\max{(2l,b)+2}$ ancilla qubits. 
$\tilde{P}_{S3}$ has a cost of
$80 \lceil\log_2{(20/\varepsilon)}\rceil + 28 b + 108 l + 20C - 12$ T gates, one unreusable ancilla qubit, and $2b$ ancilla qubits. 
Replacing ${P}_{S3}'$ on the first step and three multi controlled gates with controlled versions of them, we obtain controlled $\tilde{P}_{S3}$ at a cost of $80 \lceil\log_2{(20/\varepsilon)}\rceil + 28 b + 4\eta + 112 l + 20C + 12$ T gates, one unreusable ancilla qubit, and $2b+1$ ancilla qubits. 

\begin{figure}[htbp]
\centerline{
\includegraphics[width=125mm, page=1]{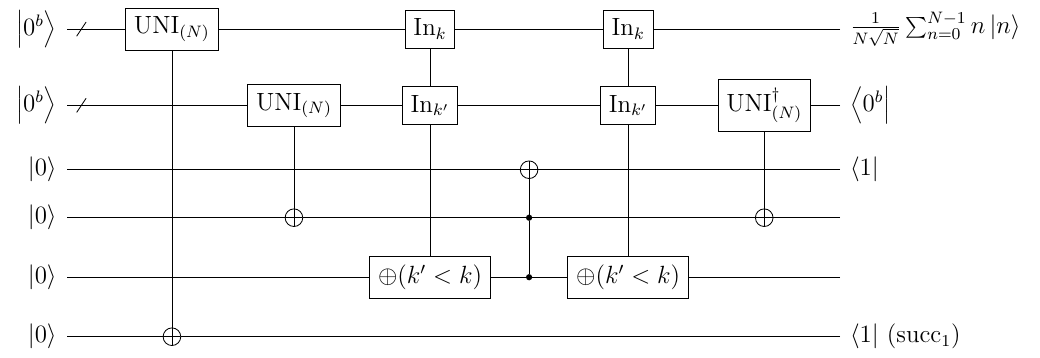}
}
\caption{Circuit for $P_{S3}'$ operator which is used as the subroutine to implement $P_{S3}$ operator.}
\label{fig:ps3prime}
\end{figure}
\begin{figure}[htbp]
\centerline{
\includegraphics[width=125mm, page=1]{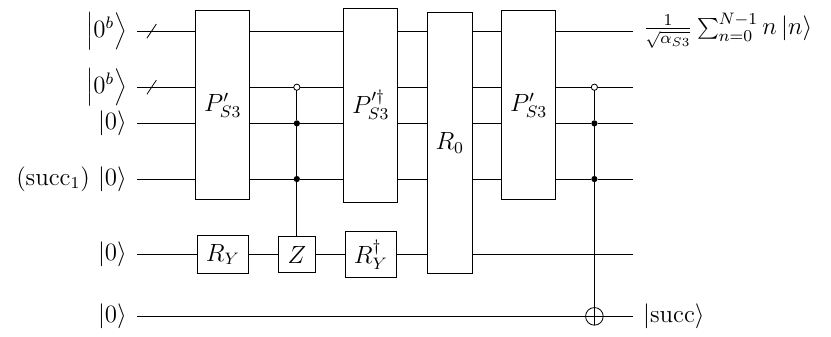}
}
\caption{Circuit for $P_{S3}$ operator.}
\label{fig:ps3}
\end{figure}
\end{proof}

\subsection{\texorpdfstring{$P_1$}{} operator}
Using Y rotations, $\mathrm{UNI}$, $P_{S1}$, $P_{S2}$, and $P_{S3}$, we can implement $P_1$ as Fig.~\ref{fig:p1}.
In the first few steps, using the appropriate Y rotation gates $R_1, R_{21}, R_{22}$, $R_3$ and a controlled Hadamard gate, we implement $P_{\rm split}$ (See Eq. \eqref{eq:split}) which generates the following state:
\begin{equation}\label{eqapp1.24}
  \begin{split}
  \frac{1}{\sqrt{\alpha_S}} &\left( \sqrt{\frac{1}{2}w\left( N-1 \right)}\ket{0}\ket{0}\ket{0} + \sqrt{\frac{1}{2}w\left( N-1 \right)} \ket{0}\ket{0}\ket{1} + \sqrt{\frac{m}{2}N} \ket{0}\ket{1}\ket{0} \right. \\
  &\left. + \sqrt{\frac{J\theta_0}{2\pi}  \alpha_{S1}} \ket{1}\ket{0}\ket{0} + \sqrt{ \left(\frac{J\theta_0}{2\pi}+\frac{ J}{2}\right)\alpha_{S2}} \ket{1}\ket{0}\ket{1} + \sqrt{\frac{ J}{8} \alpha_{S3}} \ket{1}\ket{1}\ket{0} \right).
  \end{split}
  \end{equation} 
This process is based on Ref. \cite{stateprep}.
Then, we perform the controlled versions of $\mathrm{UNI}_{(N-1)}$, $\mathrm{UNI}_{(N)}$, $P_{S1}$, $P_{S2}$, and $P_{S3}$ to obtain: 
\begin{equation}\label{eqapp1.25}
  \begin{split}
  \frac{1}{\sqrt{\alpha_S}} &\left( \sqrt{\frac{w}{2}}\ket{0}\ket{0}\ket{0} \sum_{n=0}^{N-2} \ket{n} + \sqrt{\frac{w}{2}} \ket{0}\ket{0}\ket{1} \sum_{n=0}^{N-2} \ket{n} + \sqrt{\frac{m}{2}} \ket{0}\ket{1}\ket{0} \sum_{n=0}^{N-1} \ket{n} \right. \\
  &\left. +  \sqrt{\frac{J\theta_0}{2\pi}} \ket{1}\ket{0}\ket{0} \sum_{\substack{n=1 \\ n: \mathrm{even}}}^{N-1} \sqrt{n}\ket{n} +\sqrt{\frac{J\theta_0}{2\pi}+\frac{J}{2}} \ket{1}\ket{0}\ket{1} \sum_{\substack{n=1 \\ n: \mathrm{odd}}}^{N-1} \sqrt{n}\ket{n}\right. \\
  & \left. + \sqrt{\frac{J}{8} } \ket{1}\ket{1}\ket{0} \sum_{n=1}^{N-1} n\ket{n} \right). 
  \end{split}
\end{equation}
This leads to the following result.

\begin{figure}
\centerline{
\includegraphics[width=140mm, page=1]{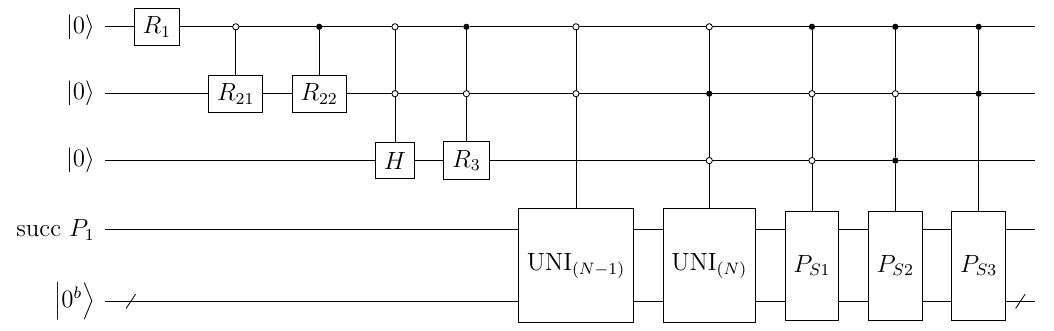}
}
\caption{Circuit for $P_1$ operator. The first three qubits are used for taking a linear combination of unitaries that encode each term in Eq.~\eqref{eq:schwinger-mod}. }
\label{fig:p1}
\end{figure}

\begin{lem}[$P_1$ operator]\label{lem:p1}
Let $N\geq 8 \in\mathbb{N}$, $b=\lceil \log_2 N \rceil$, let $C$ be a constant defined in \eqref{eq:C}, $N'=\lceil N/2\rceil$, and $N''=\lfloor N/2\rfloor$. 
Further, let $\eta'$, $L'$, $\mu'$, $K'$, $\eta''$ and $L''$ be integers such that $N' = 2^{\eta'} \cdot L', \ N'-1 = 2^{\mu'} \cdot K', and \ N'' = 2^{\eta''} \cdot L'' \ (L',K',L'': \text{odd})$, and  
$P_1$ be a unitary that generates the state \eqref{eqapp1.25} from $\ket{0}\ket{0}\ket{0}\ket{0^b}$.
We can implement a unitary $\tilde{P}_{1}$ such that
$\left\| P_{1} - \tilde{P}_{1} \right\| \leq \varepsilon$
with 
\begin{equation}
\begin{split}
    &156\lceil \log_2(39/\varepsilon) \rceil + 28 b + 19\lceil \log_2 N' \rceil + 19\lceil \log_2 N'' \rceil + 8\eta + 4\mu + 4\eta' + 4\mu' + 8\eta'' \\
    &+ 128 \lceil \log_2 L \rceil + 16 \lceil \log_2 K \rceil+ 16 \lceil \log_2 L' \rceil + 16 \lceil \log_2 K' \rceil + 32\lceil \log_2 L'' \rceil + 39 C + 104
\end{split}
\end{equation}
T gates, 
\begin{equation}
\max{(2\lceil \log_2 N' \rceil + \lceil \log_2 (N'-1) \rceil, 3\lceil \log_2 N'' \rceil)} + 4    
\end{equation}
unreusable ancilla qubits, and 
\begin{equation}
    \max{(2\lceil \log_2 L' \rceil+2\lceil \log_2 K' \rceil +1, 4\lceil \log_2 L'' \rceil +1, 2b )} + 2
\end{equation}
ancilla qubits. 
\end{lem}
\begin{proof}
Set the precision of Y rotations in the circuit (Fig. \ref{fig:p1}) to $\varepsilon/39$, which give the overall precision $\varepsilon$, and count the T gates and ancilla qubits. 
\end{proof}

\subsection{\texorpdfstring{$P_2$}{} operator}

We construct $P_2$ operator using fixed-point amplitude amplification. 
\begin{lem}[Fixed-point amplitude amplification \cite{gilyen2019quantum, yoder2014fixed, rall2023amplitude}]\label{lem:fpAA}
Let $\delta,\varDelta,a',\kappa\in(0,1)$, $\delta<\varDelta$, $a'>\kappa$ and $\Pi = \ket{\psi_0}\bra{\psi_0}$ be a projector. 
Suppose that $\tilde{\Pi}$ is a projector such that $a'\ket{\psi_G} = \tilde{\Pi}\ket{\psi_0}$.
Then there exist a unitary $U'$ such that $\tilde{\Pi}U'\ket{\psi_0}=\xi \ket{\psi_G}$, which can be implemented with $d \geq \frac{1}{\kappa}\ln{(2/\sqrt{\varDelta})} (:\text{odd})$ $C_{\Pi}NOT$, $C_{\tilde{\Pi}}NOT$, and Z rotation gates. 
Here, $\xi\in\mathbb{C}$ satisfies $|\xi|^2=(1-\delta)>(1-\varDelta)$ and $C_{\Pi}NOT=X\otimes\Pi+I\otimes(I-\Pi)$.
\end{lem}

\begin{lem}[$P_2$ operator]\label{lem:p2}
Let $N \in\mathbb{N}$, $b=\lceil \log_2 N \rceil$ and let $C$ be a constant defined in \eqref{eq:C}. 
For $n\in\{0,\dots N-1\}$, let $\delta_n,\varDelta \in(0,1)$, $\delta_n<\varDelta$, $\xi_n\in\mathbb{C}$ and $|\xi_n|^2=(1-\delta_n)>(1-\varDelta)$. 
Suppose that $P_2'$ is a unitary defined as:
\begin{equation}\label{eq:lemp2}
P_2' = \sum_{n=1}^{N-1} \ket{n}\bra{n} \otimes \left( \left(\xi_n\frac{1}{\sqrt{n}}\sum_{i=0}^{n-1}\ket{i} + \ket{\perp_n}\right)\bra{0^b} + \dots \right) + \ket{0}\bra{0}\otimes W_0 + \sum_{n=N}^{2^b-1} \ket{n}\bra{n} \otimes W_n, 
\end{equation}
where $W_0, W_N, \dots W_{2^b-1}$ are some unitaries whose actions we do not specify and $\ket{\perp_n}$ are unnormalized vectors such that $\braket{\perp_n|\perp_n}= 1-|\xi_n|^2= \delta_n<\varDelta $. 
Then we can implement $\tilde{P}_2$ such that
$\left\| P_{2}' - \tilde{P}_{2} \right\| \leq \varepsilon$
with $d(4\lceil\log_2{(d/\varepsilon)}\rceil + 8 b +C-2)+12b-6$ T gates and $6 b$ ancilla qubits, where $d\geq \sqrt{2}\ln{(2/\sqrt{\varDelta})}$ is an odd number. 
Moreover, a controlled version of $\tilde{P}_2$ can be implemented with $d(8\lceil\log_2{(2d/\varepsilon)}\rceil + 8 b +2C-2)+16b-2$ T gates and $6 b$ ancilla qubits. 
\end{lem}
\begin{proof}
A circuit for $\tilde{P}_2$ is given in Fig. \ref{fig:p2}. 
The step-by-step description of this circuit is as follows: 
\begin{enumerate}
\item Perform subtraction using out-of-place adder.
\item Produce a register which has zeros matching the leading zeros in the binary representation of $n-1$ and ones after that. This operation named $\mathrm{una}_{(n-1)}$ has a cost of $4b-4$ T gates and $2b$ ancilla qubits using out-of-place adder \cite{lee2021even}. 
\item Perform $b$ controlled Hadamard gates with $4b$ T gates. 
\item Perform fixed-point amplitude amplification with $d(4\lceil\log_2{(d/\varepsilon)}\rceil + 8 b +C-2)-4b + 2 $ T gates and $b - 1$ ancilla qubits. 
\item Flag success of the entire procedure using an inequality test. 
\item Invert $\mathrm{una}_{(n-1)}$ without T gates. 
\item Invert subtraction without T gates. 
\end{enumerate}
The register calculated by $\mathrm{una}_{(n-1)}$ in Step 2 effectively finds an integer $\lceil \log_2 n \rceil$. 
The controlled Hadamard gates at Step 3 using the information from the previous step generates a state that is sufficiently close to $(1/\sqrt{n})\sum_{i=0}^{n-1}\ket{i}$ for each $n$, that is, a state that has overlap larger than $1/\sqrt{2}$ with $(1/\sqrt{n})\sum_{i=0}^{n-1}\ket{i}$.
Therefore, the fixed-point amplitude amplification with iteration $d\geq \sqrt{2}\ln{(2/\sqrt{\varDelta})}$ can produce the operation in \eqref{eq:lemp2} with the error parameter $\Delta$ using Lemma \ref{lem:fpAA}.

The overall cost of $P_2'$ is $d(4\lceil\log_2{(d/\varepsilon)}\rceil + 8 b +C-2)+12b-6$ T gates and $6 b$
ancilla qubits. 
Furthermore, we obtain a controlled version of $\tilde{P}_2$ by replacing controlled Hadamard gates on the third step, Z rotations during amplitude amplification, and CNOT gate on the fifth step with controlled versions of them. 
Controlled $\tilde{P}_2$ requires $d(8\lceil\log_2{(2d/\varepsilon)}\rceil + 8 b +2C-2)+16b-2$ T gates and $6 b$ ancilla qubits. 
\end{proof}

\begin{figure}[htbp]
\centerline{
\includegraphics[width=180mm, page=1]{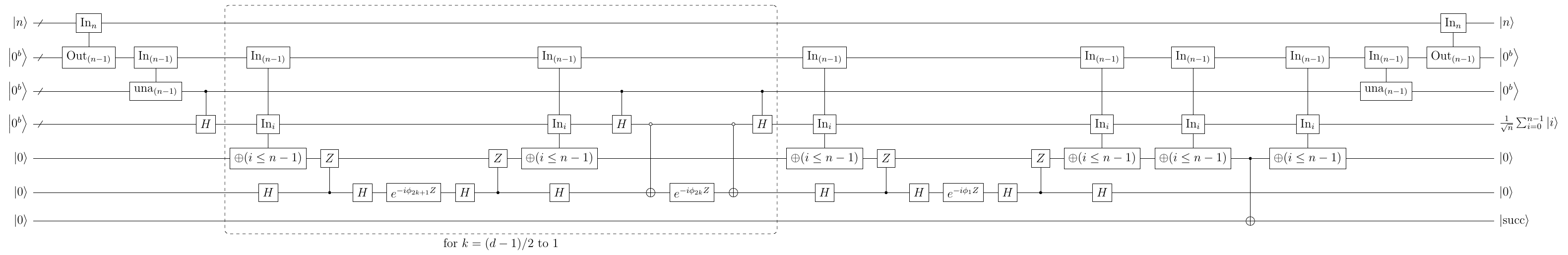}
}
\caption{Circuit for $P_2$ operator. }
\label{fig:p2}
\end{figure}

\section{Resource estimates}\label{appsec:resource-estimates}

Here, using the subroutines constructed in Appendix \ref{appsec:subroutines}, we prove Results~\ref{result:schwinger-block-encoding:main}-\ref{cor:vacuum-amplitude-estimate-main} in the main text.
The discussion that follows is summarized as Table \ref{tab:algo}.

\begin{table}
    \centering
    \scriptsize
\caption{Cost of high-level algorithms. $N$ denotes the system size.\label{tab:algo}}
\renewcommand{\arraystretch}{2}
    \begin{tabular}{p{2cm}p{0.6cm}p{3.2cm}p{2.5cm}p{4.5cm}}\hline\hline
         Algorithm &  Param &  T Cost & \# of ancilla& Note\\\hline
          Block-encoding of $H_S$ ($U_{H_S}$)& $\varepsilon$ & \raggedright$4\times \text{c}P_{2}(N, \frac{\varepsilon}{14\alpha_S}, \frac{\varepsilon}{14\alpha_S})+2\times P_1(N, \frac{\varepsilon}{14\alpha_S})+\text{c3-}V_{XX}+\text{c3-}V_{YY}+\text{c2-}V_{Z}+\text{c3-}V_{Z2}+\text{c4-}V_{Z2}+R_0(b+3)$ & $P_2(N)+P_1(N,\text{unreusable})+2$ & Normalization constant is $\alpha_S = w(N-1) + \frac{m}{2}N +  \frac{J\theta}{2\pi}\sum_{l=1, l: \mathrm{even}}^{N-1} l +  \left(\frac{J\theta}{2\pi}+\frac{J}{2}\right)\sum_{l=1, l: \mathrm{odd}}^{N-1} l + \frac{J}{8}\sum_{l=1}^{N-1} l^2$. \\
 Block-encoding of $e^{-iH_St}$ ($U_{\text{time}}$)& $\varepsilon,t$&\raggedright $3r U_{H_S}(\frac{\varepsilon}{3t})+ 6(2r+1)R_Y(\frac{\varepsilon}{18(2r+1)}) + 3(r+1)R_0(2b+4) + 3\text{c-}U_{H_S}(\frac{\varepsilon}{3t}) + 2R_0(2b+5)$  & $U_{H_S}(N,\frac{\varepsilon}{3t})$ &  $r$ is the smallest even number $\geq 2\alpha_S |t| + 3\ln{(9/\varepsilon)}$. $b=\lceil\log_2 N\rceil$.\\
 Estimating vacuum persistence amplitude & $\varepsilon, t$& $4Q_{\text{est}}U_{\text{time}}(t,\varepsilon/2)+2Q_{\text{est}}R_0(2b+5+N)+2Q_{\text{est}}R_0(2b+6+N)$& $\text{max}(R_0(2b+6+N), U_{\text{time}}(N))$ &$Q_{\mathrm{est}} =  \left\lceil\frac{5.874534}{\varepsilon} \ln{\left( 2.08 \ln{\left( \frac{2}{\varepsilon} \right)} \right)} \right\rceil $, $b=\lceil\log_2 N\rceil$.\\
    \hline \hline
    \end{tabular}
\end{table}

\subsection{Resource estimates for block-encoding}\label{appsubsec:resource-block}

Here, we prove Result \ref{result:schwinger-block-encoding:main}. For readers' convenience, we restate the result as follows:

\setcounter{result}{0}

\begin{result}[Block-encoding of the Schwinger model Hamiltonian]
Let $N\geq8\in\mathbb{N}$ be a system size and let $H_S$ be the Schwinger model Hamiltonian given by Eq. (\ref{eq:schwinger}). 
Let $b=\lceil\log_2{N}\rceil$.
For any $\varepsilon>0$, we can implement an unitary $U_{H_S}$ which is an $(\alpha_S, 2b+3, \varepsilon)$-block-encoding of 
\begin{equation}
H_{S,\mathrm{mod}}\equiv
H_S - \frac{J}{8}\sum_{l=1}^{N-1} l^2 - J\sum_{n=1}^{N-1}\left(\frac{1}{2}\frac{1+(-1)^{n-1}}{2}+\frac{\theta}{2\pi}\right)^2 ,
\end{equation}
where
\begin{equation}
\alpha_S = w(N-1) + \frac{m}{2}N +  \frac{J\theta}{2\pi} \sum_{\substack{l=1 \\ l: \mathrm{even}}}^{N-1} l +\left(\frac{J\theta}{2\pi}+\frac{J}{2}\right) \sum_{\substack{l=1 \\ l: \mathrm{odd}}}^{N-1} l + \frac{J}{8}\sum_{l=1}^{N-1} l^2,
\end{equation}
using the cost listed in Table \ref{tab:summary-block-encoding}.
\end{result}
\begin{proof}
We have seen in the main text that the circuit in Fig. \ref{fig:schwinger-block} is a $(\alpha_S, 2b+3, 0)$-block-encoding of $H_{S, \text{mod}}$ with ideal operations $P_1$ and $P_2$.
We therefore analyze the error introduced by using approximate $P_1$ and $P_2$ below.

Let $P_1$, $P_2$, $P_2'$ be unitaries defined in Lemmas \ref{lem:p1} and \ref{lem:p2}. 
We can construct a unitary $\tilde{P}_1, \tilde{P}_2$ such that $\|P_1-\tilde{P}_1\|\leq \varepsilon_1$, $\|P_2'-\tilde{P}_2\|\leq \varepsilon_2$ from Lemmas \ref{lem:p1} and \ref{lem:p2}.
Let $U_{H_S}$ be the exact block-encoding of the Schwinger model Hamiltonian in Fig. \ref{fig:schwinger-block}, $U_{H_S}'$ be its approximation using $P_2'$ instead of the ideal $P_2$, $\tilde{U}_{H_S}'$ be its approximation using $\tilde{P}_2$ instead of the ideal $P_2$, and $\tilde{U}_{H_S}$ be its approximation using $\tilde{P}_1$ and $\tilde{P}_2$ instead of the ideal $P_1$ and $P_2$. 
$\tilde{U}_{H_S}$ is the unitary that is actually implemented.
First, we can observe that
\begin{equation}\label{eq:thm1:first-appro}
\|\tilde{U}_{H_S}' - \tilde{U}_{H_S} \| \leq 2\varepsilon_1
\end{equation}
since there are two $P_1$'s in the circuit that are approximated by $\tilde{P}_1$.
Next, we can also observe that
\begin{equation}\label{eq:thm1:second-appro}
\| U_{H_S}' - \tilde{U}_{H_S}' \| \leq 4\varepsilon_2 
\end{equation}
since there are four $P_2'$'s that are approximated by $\tilde{P}_2$.

Finally, we evaluate $\|(\bra{0}\otimes I) U_{H_S} (\ket{0}\otimes I) - (\bra{0}\otimes I) U_{H_S}' (\ket{0}\otimes I) \|$. 
$U_{H_S}'$ block-encodes $H_{XX},H_{YY},H_{Z}$ exactly since it uses the ideal $P_1$. 
Therefore, we evaluate the error of $H_{Z\text{even}}, H_{Z\text{odd}}, H_{Z^2}$. 
First, observe that $P_1$ generates a state like, 
\begin{equation}\label{eq:thm1:p1-def}
P_1 \ket{0^b} = \left( \sqrt{p_{S1}} \frac{1}{\sqrt{\alpha_{S1}}} \sum_{\substack{n=1 \\ n: \mathrm{even}}}^{N-1} \sqrt{n}\ket{n} + \sqrt{p_{S2}} \frac{1}{\sqrt{\alpha_{S2}}} \sum_{\substack{n=1 \\ n: \mathrm{odd}}}^{N-1} \sqrt{n}\ket{n}  + \sqrt{p_{S3}} \frac{1}{\sqrt{\alpha_{S3}}} \sum_{n=1}^{N-1} n\ket{n} + \dots \right),
\end{equation}
where we omit the three ancilla qubits. 
The above equation, Eq. \eqref{eq:thm1:p1-def}, should be interpreted as the one that defines $p_{S1}$, $p_{S2}$ and $p_{S3}$ to match the definition of $P_1$, Eq. \eqref{eqapp1.25}.
$U_{H_S}$ approximately encodes $H_{Z\text{even}}$, $H_{Z\text{odd}}$ and $H_{Z^2}$ as,
\begin{align}\label{eq:thm1:Zeven}
(\bra{0^b}\bra{0^b})P_1^\dag \tilde{P}_2^\dag V_{Z^2} \tilde{P}_2 P_1 (\ket{0^b}\ket{0^b})
&= \frac{p_{S1}}{\alpha_{S1}}  \sum_{\substack{n=1 \\ n: \mathrm{even}}}^{N-1} n \left((1-\delta_n)\frac{1}{n}\sum_{i=0}^{n-1} Z_i + \delta_n \right), \\ 
(\bra{0^b}\bra{0^b})P_1^\dag \tilde{P}_2^\dag V_{Z^2} \tilde{P}_2 P_1 (\ket{0^b}\ket{0^b})
&= \frac{p_{S2}}{\alpha_{S2}}  \sum_{\substack{n=1 \\ n: \mathrm{odd}}}^{N-1} n \left((1-\delta_n) \frac{1}{n} \sum_{i=0}^{n-1} Z_i + \delta_n \right), \label{eq:thm1:Zodd} \\
\label{eq:thm1:Z2}
%\begin{split}
    (\bra{0^b}\bra{0^b})P_1^\dag \tilde{P}_2^\dag V_{Z^2} \tilde{P}_2 (I\otimes R_0)) \tilde{P}_2^\dag V_{Z^2} \tilde{P}_2 P_1 (\ket{0^b}\ket{0^b})
&= \frac{2 p_{S3}}{\alpha_{S3}} \sum_{n=1}^{N-1} n^2 \left( (1-\delta_n) \frac{1}{n} \sum_{i=0}^{n-1} Z_i + \delta_n \right)^2 - p_{S3}
%\end{split}
\end{align}
where $\delta_n$ is defined in Lemma \ref{lem:p2}. 
As a result, 
\begin{align}\label{eq:thm1:third appro}
&\|(\bra{0}\otimes I) U_{H_S} (\ket{0}\otimes I) - (\bra{0}\otimes I) U_{H_S}' (\ket{0}\otimes I) \| \\
\begin{split}\nonumber
=& \left\| \frac{p_{S1}}{\alpha_{S1}}  \sum_{\substack{n=1 \\ n: \mathrm{even}}}^{N-1} n \left(\delta_n\frac{1}{n}\sum_{i=0}^{n-1} Z_i - \delta_n\right) + \frac{p_{S2}}{\alpha_{S2}}  \sum_{\substack{n=1 \\ n: \mathrm{odd}}}^{N-1} n \left(\delta_n \frac{1}{n} \sum_{i=0}^{n-1} Z_i - \delta_n\right) \right. \\
&\left. 
+ \frac{2 p_{S3}}{\alpha_{S3}} \sum_{n=1}^{N-1} n^2 \left( \frac{1}{n} \sum_{i=0}^{n-1} Z_i \right)^2
- \frac{2 p_{S3}}{\alpha_{S3}} \sum_{n=1}^{N-1} n^2 \left( (1-\delta_n) \frac{1}{n} \sum_{i=0}^{n-1} Z_i + \delta_n \right)^2
\right\| 
\end{split}\\ \nonumber
\leq& 2 p_{S1} \varDelta + 2 p_{S2} \varDelta + \frac{2 p_{S3}}{\alpha_{S3}} \sum_{n=1}^{N-1} n^2  \left\|  \left( \frac{1}{n} \sum_{i=0}^{n-1} Z_i \right)^2 - \left( (1-\delta_n) \frac{1}{n} \sum_{i=0}^{n-1} Z_i + \delta_n \right)^2
\right\| \\ \nonumber
\leq &  2 p_{S1} \varDelta + 2 p_{S2} \varDelta + \frac{2 p_{S3}}{\alpha_{S3}} \sum_{n=1}^{N-1} n^2 \left(
\left\| \frac{1}{n} \sum_{i=0}^{n-1} Z_i \right\| \cdot
\left\| \delta_n \frac{1}{n} \sum_{i=0}^{n-1} Z_i - \delta_n \right\|  
+ 
\left\| \delta_n \frac{1}{n} \sum_{i=0}^{n-1} Z_i - \delta_n \right\| \cdot
\left\| (1-\delta_n) \frac{1}{n} \sum_{i=0}^{n-1} Z_i + \delta_n \right\| 
\right) \\ \nonumber
\leq &  2 p_{S1} \varDelta + 2 p_{S2} \varDelta + \frac{2 p_{S3}}{\alpha_{S3}} \sum_{n=1}^{N-1} n^2 \left( 2\varDelta + 2\varDelta \right) \\ 
\leq & 2 p_{S1} \varDelta + 2 p_{S2} \varDelta + 8 p_{S3} \varDelta 
\leq 8 \varDelta, 
\end{align}
where in the last inequality we use $p_{S1}+p_{S2}+p_{S3}\leq1$. 
We therefore have,
\begin{align}\label{eq:thm1:fourth-appro}
\left\| (\bra{0}\otimes I) U_{H_S} (\ket{0}\otimes I) - (\bra{0}\otimes I) \tilde{U}_{H_S} (\ket{0}\otimes I) \right\| 
\leq 2\varepsilon_1 + 4\varepsilon_2 + 8 \varDelta. 
\end{align}
Setting $\varepsilon_1 = \varepsilon_2 = \varDelta = \varepsilon/(14 \alpha_S)$, we get the $(\alpha_S, 2m+3, \varepsilon)$-block-encoding of $H_{S,\mathrm{mod}}$: 
\begin{equation}\label{eq:thm1:fifth-appro}
\left\| H_{S,\mathrm{mod}} - \alpha_S(\bra{0}\otimes I) \tilde{U}_{H_S} (\ket{0}\otimes I) \right\|
\leq \varepsilon. 
\end{equation}
Counting T gates and ancilla qubits with the above error parameters, we obtain the stated result. 
\end{proof}

\subsection{Resource estimates for Hamiltonian simulation and amplitude estimation}\label{appsubsec:hamsim-ampest}

\subsubsection{Hamiltonian simulation}

We use the following Lemma that gives the detailed cost for Hamiltonian simulation using its block-encoding. 
Note that the Z rotations used in the Lemma is assumed to be exact ones without error.

\begin{lem}[Robust block-Hamiltonian simulation \cite{gilyen2019quantum, berry2014exponential}]\label{lem:hamiltonian-simulation-resource}
Let $t\in\mathbb{R}$ and $\varepsilon\in(0,1)$. 
Suppose that $U_H$ is an $(\alpha,l,\varepsilon/2|t|)$-block-encoding a Hamiltonian $H$ and $r'$ is the smallest even number $\geq 2\alpha|t| + 3\ln{(6/\varepsilon)}$. 
Then we can implement a unitary which is an $(1,l+2,\varepsilon)$-block-encoding of $e^{-iHt}$ with $3r'$ uses of $U_H$ or its inverse, $3$ uses of controlled $U_H$ or its inverse, $6r+3$ uses of controlled Z rotation gates, $3(r+1)$ uses of reflection operators $2\ket{0^{l+1}}\bra{0^{l+1}}-I$, and $2$ uses of reflection operators $2\ket{0^{l+2}}\bra{0^{l+2}}-I$. We provide the circuit to realize this in Fig.~\ref{fig:Hamiltonian_simulation}.
\end{lem}

\begin{figure}
  \centering
  \subfloat[]{
    \includegraphics[width=1\textwidth]{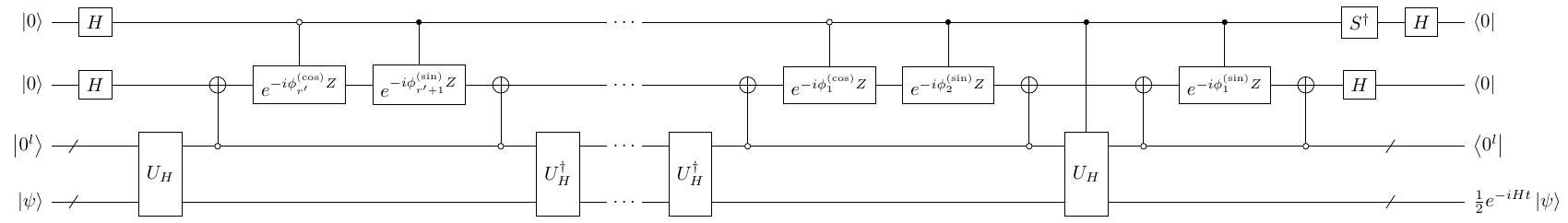}
    \label{fig:Hamiltonian_simulation_step1}
  }
  \\ 
  \subfloat[]{
    \includegraphics[width=0.4\textwidth]{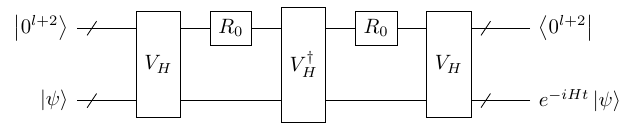}\hfill
    \label{fig:oblivious_amplitude_amplification}
  }
  \caption{\label{fig:Hamiltonian_simulation}Circuit for robust block-Hamiltonian simulation (Lemma \ref{lem:hamiltonian-simulation-resource}). (a) Circuit for realizing \(\frac{1}{2}e^{-iHt}\). $\phi_{j}^{\rm (sin)}$ and $\phi_{j}^{\rm (cos)}$ are real parameters predetermined by quantum signal processing technique \cite{gilyen2019quantum}. (b) Circuit for realizing $e^{-iHt}$ via oblivious amplitude amplification. $V_H$ denotes the whole circuit of (a).}
  \label{fig:combined_figures}
\end{figure}

We now restate the cost for simulatng the Schwinger model Hamiltonian which is Result \ref{cor:simulation-schwinger-main} in the main text, and give its proof. 
\setcounter{cor}{0}
\begin{result}[Simulating the Schwinger model Hamiltonian]\label{cor:simulation-schwinger}
Let $t\in \mathbb{R}$, $\varepsilon\in(0,1)$, $N\geq 8 \in \mathbb{N}$ be a system size, $b=\lceil \log_2{N} \rceil$, and $r$ be the smallest even number $\geq 2\alpha_S |t| + 3\ln{(9/\varepsilon)}$. 
Let $\alpha_S$, $C$ and $H_{S,\mathrm{mod}}$ be as defined in Result~\ref{result:schwinger-block-encoding:main}.
Suppose that $U_{H_S}$ is a unitary which is an $(\alpha_S,2b+3,\varepsilon/(3|t|))$-block-encoding of $H_{S,\mathrm{mod}}$ and let $C_{H_S}$ be the number of T gates required to implement $U_{H_S}$ given in Result \ref{result:schwinger-block-encoding:main}. 
We can implement a unitary $U_{\mathrm{time}}$ which is an $(1,2b+5,\varepsilon)$-block-encoding of $e^{-i H_{S,\mathrm{mod}} t}$ using 
\begin{equation}
 r \left( 3C_{H_S} + 48 \left\lceil\log_2{\left(\frac{18(2r+1)}{\varepsilon} \right)}\right\rceil + 24 b + 12 C + 24 \right) \\
+ 3C_{H_S} + 24 \left\lceil\log_2{\left(\frac{18(2r+1)}{\varepsilon} \right)}\right\rceil + 40 b + 6C + 120
 \end{equation}
T gates and
\begin{equation}
 6b + \max{(2\lceil \log_2 N' \rceil + \lceil \log_2 (N'-1) \rceil, 3\lceil \log_2 N'' \rceil)} + 6
\end{equation}
ancilla qubits.
\end{result}
\begin{proof}
We use Lemma \ref{lem:hamiltonian-simulation-resource} to obtain $(1,2b+5,2\varepsilon/3)$-block-encoding of $e^{-i H_{S,\mathrm{mod}} t}$ with exact Z rotations, and set the error caused by approximate Z rotations due to the Clifford+T decomposition as $\varepsilon/3$ to achieve $\varepsilon$-precision for the overall procedure. Using these error parameters and counting the number of T gates and ancilla qubits, we obtain the stated result.
Note that controlled $U_{H_S}$ can be implemented using $C_{H_S}+24$ T gates by adding another control qubit for every SELECT circuit in Fig.~\ref{fig:schwinger-block}. 
\end{proof}

\subsubsection{Amplitude estimation}

Next, we provide the cost analysis for Chebyshev amplitude estimation.
Let $\ket{\psi}$ be a state, $\Pi$ be a projector, and $\omega=\|\Pi \ket{\psi}\|$. 
The objective of Chebyshev amplitude estimation is to construct a random variable $\hat{\omega}$ satisfying $\mathrm{Pr}[|\hat{\omega}-\omega|\geq \varepsilon]\leq \delta$ by calling reflection operators $R_{\Pi} = 2\Pi - I$ and $R_{\psi} = 2\ket{\psi}\bra{\psi}-I$.
Rough skecth of the algorithm is as follows: for some integer $d$, apply $\cdots R_{\Pi} R_{\psi} R_{\Pi}$ to $\ket{\psi}$ in such a way that $R_{\Pi}$ and $R_{\psi}$ are applied $d$ times in total. Then, at the end of the circuit, measure in the eigenbasis of $R_{\Pi}$ or $R_{\psi}$. 
This can be performed with e.g. the Hadamard test using the controlled-$R_\Pi$ or controlled-$R_{\psi}$. 
We repeat the above process appropriately increasing $d$ based on the past measurement results to get an estimate $\hat{\omega}$.

Let us denote the numbers of queries to $2\Pi - I$ and $2\ket{\psi}\bra{\psi}-I$, including its controlled versions for measurements, throughout the algorithm by $Q_{\Pi}$ and $Q_{\psi}$ respectively.
Note that the controlled versions of $2\Pi - I$ and $2\ket{\psi}\bra{\psi}-I$ in our problem setting can be realized with merely four more T gates than their original versions.
We can safely neglect this additional cost to realize the controlled reflections as the overall T count to implement the uncontrolled version of $2\ket{\psi}\bra{\psi}-I$ is over $10^8$ even for relatively small system sizes (see Fig. \ref{fig:comp-T-overall} in the main text). 
As $d$ is determined from the random process, that is, the history of the measurement, $Q_{\Pi}$ and $Q_{\psi}$ varies probabilistically.
Their distribution depend on $\omega$, $\varepsilon$, and $\delta$.
Note that it has been observed that $Q_{\Pi} \approx Q_{\psi}$ (see the proof of Proposition 3 of Ref. \cite{rall2023amplitude}), as can be intuitively speculated from the structure of the algorithm.

To evaluate $Q_{\Pi}$ and $Q_{\psi}$ for our specific problem setting, we numerically simulate the algorithm based on the program code provided in Ref.~\cite{rall2023amplitude} for $\varepsilon = 0.005$ and $\delta = 0.05$.
In Fig. \ref{fig:chebae_analysis}, we show the distribution of $Q_{\psi}+Q_{\Pi}$ of 1000 trials each for different $\omega$ in this error setting.
We can observe that we need $Q_{\psi}+Q_{\Pi}\approx 2000$ on average for most challenging cases (small $\omega$'s) and less than 2000 when $\omega$ is larger.
Combined with the emprical property of Chebyshev amplitude estimation that $Q_{\Pi} \approx Q_{\psi}$, we conclude that $Q_{\psi}=Q_{\Pi}=1000$ can achieve $\varepsilon = 0.005$ and $\delta = 0.05$ on average.

Let us now argue the correctness of Result \ref{cor:vacuum-amplitude-estimate-main}. For convenience, we restate it below:

\begin{result}[Estimating the vacuum persistence amplitude based on empirical assumption]\label{cor:vacuum-amplitude-estimate}
Let $t\in \mathbb{R}$, $N\geq 8 \in \mathbb{N}$ be a system size, $b=\lceil \log_2{N} \rceil$, $\ket{\mathrm{vac}} =\ket{1 0 1 0 \cdots }$, and
$C_{\mathrm{time}}$ be the number of T gates required to implement a $(1,2b+5,0.005)$-block-encoding of $e^{-i H_{S,\mathrm{mod}} t}$ via Result \ref{cor:simulation-schwinger-main}. 
Then we can sample from a random variable $\hat{G}$ satisfying $\mathrm{Pr}[|\hat{G}- |\braket{\mathrm{vac}|e^{-iH_S t}|\mathrm{vac}}| |\geq 0.01]\leq 0.05 $ using about $2000(C_{\mathrm{time}} + 4N + 8b + 12)$ T gates on average and $\max(N+2b+3, 6b + \max(2\lceil \log_2 N' \rceil + \lceil \log_2 (N'-1) \rceil, 3\lceil \log_2 N'' \rceil) + 6)$ ancilla qubits.
\end{result}

As disscussed in the main text, to obtain the vaccum persistence amplitude,
we set $\Pi = \ket{\text{vac}'}\bra{\text{vac}'}$, where $\ket{\text{vac}'} = \ket{0^{\nu -N}}\ket{\text{vac}}$ and $\nu = N+2\lceil \log_2{N} \rceil+5$, and implement $2\ket{\psi}\bra{\psi}-I$ via $U_{\text{time}}(2\ket{\text{vac}'}\bra{\text{vac}'}-I)U_{\text{time}}^\dagger$, where $U_{\text{time}}$ is the unitary which block-encodes $e^{-iH_{S,\text{mod}}t}$.
We therefore need two calls to $U_{\text{time}}$ and a single call of $2\ket{\text{vac}'}\bra{\text{vac}'}-I$ to implement $2\ket{\psi}\bra{\psi}-I$.
$2\ket{\text{vac}'}\bra{\text{vac}'}-I$ is equivalent to $2\ket{0^\nu}\bra{0^\nu}-I$ up to X gates and $2\ket{0^\nu}\bra{0^\nu}-I$ can be implemented with $4\nu - 8$ T gates.
We use the Chebyshev amplitude estimation with $Q_{\psi}=Q_{\Pi}=1000$ on average to achieve $\varepsilon=0.005$ and $\delta=0.05$.
This argument shows the correctness of the T count.

Next, we need $\nu-2$ ancilla qubits to implement $2\ket{0^{\nu}}\bra{0^{\nu}}-I$ with $4\nu-8$ T gates. We therefore see the correctness of the statement about ancilla qubits.

\begin{figure}
\centerline{
\includegraphics[width=125mm, page=1]{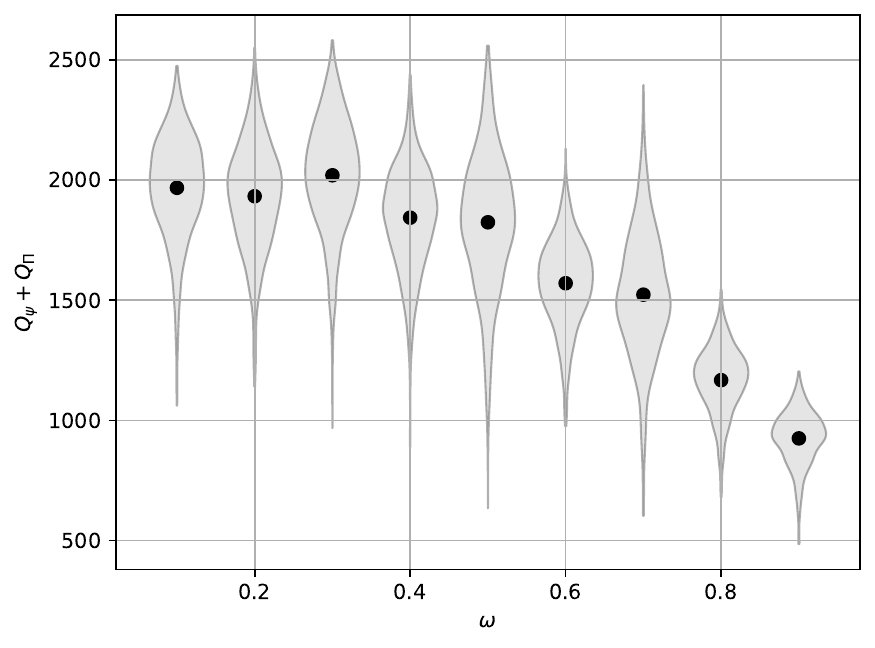}
}
\caption{The analysis of query complexity for Chebyshev amplitude estimation. $Q_\psi$ denotes the number of queries to the reflection $2\ket{\psi}\bra{\psi}-I$ and $Q_\Pi$ denotes the number of queries to the reflection $2\Pi-I$. Each violin has 1000 samples and the dots indicate the means. }
\label{fig:chebae_analysis}
\end{figure}

Finally, we provide error analysis.
First, note that $|\braket{\mathrm{vac}|e^{-iH_S t}|\mathrm{vac}}| = |\braket{\mathrm{vac}|e^{-i H_{S,\mathrm{mod}} t}|\mathrm{vac}}|$. 
By using $Q_{\psi} = Q_\Pi =1000$ queries on average, we can obtain an estimate $\hat{G}$ such that 
$|\hat{G}- | \bra{\text{vac}'} U_{\mathrm{time}} \ket{\text{vac}'} | |\leq 0.005$.
This $\hat{G}$ satisfies the following inequalities:
\begin{align*}
& \left|\hat{G}- |\braket{\mathrm{vac}|e^{-iH_S t}|\mathrm{vac}}| \right| \\
& \leq \left|\hat{G}- | \bra{\text{vac}'} U_{\mathrm{time}} \ket{\text{vac}'}  | \right| + \left| | \bra{\text{vac}'} U_{\mathrm{time}} \ket{\text{vac}'}  | - |\braket{\mathrm{vac}|e^{-i H_{S,\mathrm{mod}} t}|\mathrm{vac}}| \right| \\
& \leq 0.005 + \left| \bra{0^{2b+5}}\bra{\text{vac}} U_{\mathrm{time}} \ket{0^{2b+5}}\ket{\text{vac}} - \braket{\mathrm{vac}|e^{-i H_{S,\mathrm{mod}} t}|\mathrm{vac}} \right| \\
&= 0.005 + \left| \bra{0^{2b+5}} U_{\mathrm{time}} \ket{0^{2b+5}} - e^{-i H_{S,\mathrm{mod}} t} \right| \\
&\leq 0.005 + 0.005 = 0.01. 
\end{align*}
This concludes the argument to show the correctness of Result \ref{cor:vacuum-amplitude-estimate-main}.

In the above process, we used Chebyshev amplitude estimation to estimate the vacuum persistent amplitude. 
Here, we would like to provide comparison between Chebyshev amplitude estimation and simpler methods (such as the Hadamard test or direct measurement). 
In the case of the Hadamard test, the number of queries to the time evolution operator is given by Hoeffding's inequality:
Let $X_1, \dots , X_Q$ be independent random variables such that $X_q \in [0,1]$ and $\bar{X} = \frac{1}{Q}\sum_{q=1}^{Q}X_q$. Then, for any $\varepsilon \geq 1$, 
\begin{equation*}
\mathrm{P}\left( \left| \bar{X} - \mathrm{E}\left( \bar{X} \right) \right| \geq \varepsilon \right) < 2 e^{-2Q \varepsilon^2}, 
\end{equation*}
where $\mathrm{E}\left( \bar{X} \right)$ is the expected value of $\bar{X}$. 
To achieve $\varepsilon = 0.005$ and $\delta = 0.05$, we have to set the number of queries $Q$ as follows: 
\begin{align*}
2 e^{-2Q \varepsilon^2} &< \delta \\
Q &> \frac{1}{2 \varepsilon^2} \ln{\left( \frac{2}{\delta} \right)} \approx 73778. 
\end{align*}
On the other hand, in the case of Chebyshev amplitude estimation, the number of queries was about 2000. 
Therefore, even considering constant factors, Chebyshev amplitude estimation can be said to be better.

\end{document}